%% file: main.tex
\documentclass{elsarticle}

\biboptions{authoryear}

\let\ElseVierBibliography\bibliography%
\renewcommand{\bibliography}[1]{%
\ElseVierBibliography{#1}%
}%

\input{newcom}

\date{}
\journal{Nonlinear Analysis: Hybrid Systems}
\makeatletter
\def\ps@pprintTitle{\let\@oddhead\@empty\let\@evenhead\@empty\let\@oddfoot\@empty\let\@evenfoot\@empty}
\makeatother
\begin{document}

\begin{frontmatter}

\title{Symbolic Control for Stochastic Systems\\ via Finite Parity Games}

\author[mpisws]{Rupak Majumdar}
\ead{rupak@mpi-sws.org}

\author[ista]{Kaushik Mallik\corref{mycorrespondingauthor}}
\ead{kmallik@mpi-sws.org}

\author[mpisws]{Anne-Kathrin Schmuck\corref{mycorrespondingauthor}}
\cortext[mycorrespondingauthor]{Corresponding authors}
\ead{akschmuck@mpi-sws.org}

\author[newcastle]{Sadegh~Soudjani}
\ead{sadegh.soudjani@ncl.ac.uk}

\address[mpisws]{Max Planck Institute for Software Systems, Kaiserslautern, 67663 Germany}
\address[ista]{Institute of Science and Technology Austria (ISTA), Klosterneuburg, 3400 Austria}
\address[newcastle]{School of Computing, Newcastle University, NE4 5TG United Kingdom}

\input{abstract}


\end{frontmatter}

\input{intro}

\input{prelims}

\input{two-and-half-game}

\input{soundness}

\input{example}

\new{
\section{Conclusions}

We presented an approach to find approximately optimal controllers for continuous stochastic systems with $\omega$-regular control objectives.
We used an abstraction-based approach, where the problem is approximated by the problem of computing the winning strategy in a finite $\twohalf$-player game.
By plugging in our recent advanced symbolic algorithms for $\twohalf$-player games, we obtained an efficient symbolic procedure for solving the original controller synthesis problem.
Our approach is significantly faster than the state of the art---both in theory and in our empirical studies.

Several future directions exist to improve our results.
Firstly, the exact characterization of the optimal controller, maximizing the probability of satisfaction of an $\omega$-regular objective by a continuous stochastic nonlinear system, is still open.
This is an important theoretical question, whose answer will help us to understand how conservative the under-approximation of the winning domain is due to our choice of memoryless controllers, and may possibly help us to improve the quality of the approximation by using a more appropriate type of controller in the first place.
Secondly, for systems with large state space dimensions, our abstraction based on uniform gridding may become a computational bottleneck.
Adaptive abstraction-refinement techniques, like in other abstraction-based control literature \citep{Coogan20,hsu2018multi,nilsson2017augmented,SA13}, will help us to improve the scalability of our method.
}

\medskip
\new{
	\noindent\textbf{Acknowledgments:}
		We thank Daniel Hausmann and Nir Piterman for their valuable comments on an earlier version of the manuscript of our other paper \citep{banerjee2023fast} where we present, among other things, the parity fixpoint for $\twohalf$-player games (for a slightly more general class of games) with a different and indirect proof of correctness. 
		Based on their comments we observed that, unlike the other fixpoints that we present in \citep{banerjee2023fast}, the parity fixpoint does not follow the exact same structure as its counterpart for $2$-player games, which we also use int his paper.
    
		We also thank Thejaswini Raghavan for observing that our symbolic parity fixpoint algorithm can be solved in quasi-polynomial time using recent improved algorithms for solving $\mu$-calculus expressions. This significantly improved the complexity bounds of our algorithm in this paper.
	}

\bibliography{reportbib}

\input{appendix}

\end{document}

%% file: newcom.tex
\usepackage{amsmath,amssymb,amsfonts}
\usepackage{graphicx}
\usepackage{textcomp}
\usepackage[binary-units=true]{siunitx}
\usepackage{comment}
\usepackage{amsmath}
\usepackage{amsthm}
\usepackage{etoolbox}
\usepackage{color,colortbl}
\usepackage{mathtools}
\usepackage{enumerate}
\usepackage{algorithm,algorithmicx,algpseudocode}
\usepackage{booktabs}
\usepackage[normalem]{ulem}
\usepackage{pgf}
\usepackage{pgfplotstable}
\usepackage{tikz,pgfplots}
\usetikzlibrary{trees,decorations,arrows,arrows.meta,automata,shadows,positioning,plotmarks,backgrounds,shapes,shapes.misc}
\usetikzlibrary{calc,matrix,fit,petri,decorations.pathmorphing,patterns}
\usetikzlibrary{decorations.pathreplacing,decorations.markings,shapes.geometric,calc}
\usepackage{paralist}
\usepackage{stmaryrd}
\usepackage{xspace}
\usepackage{graphicx}
\usepackage{float}
\usepackage[utf8]{inputenc}
  \usepackage{csquotes}
 \usepackage{subfig}
\usepackage{multirow}
\usepackage{array}
\usepackage{dsfont}
 \usepackage{array}
\usepackage{xifthen}
\usepackage{relsize}
\usepackage{xfrac}
\usepackage{wrapfig}
\usepackage{rotating}

\newtheorem{theorem}{Theorem}[section]
\newtheorem{lemma}[theorem]{Lemma}

\newtheorem{proposition}[theorem]{Proposition}

\newtheorem{definition}[theorem]{Definition}
\newtheorem{example}[theorem]{Example}

\newcommand{\new}[1]{{\color{black} #1}}

\newtheorem{problem}{Problem}

\newcommand{\Ker}{T_{\mathfrak s}}

\newcommand{\ineven}{\in_{\mathit{even}}}
\newcommand{\inodd}{\in_{\mathit{odd}}}
\newcommand{\polylog}{\mathrm{polylog}}

\newcommand{\p}[1]{\mathit{Player}~#1}

\newcommand{\game}{\mathcal{G}}
\newcommand{\Bh}{\widehat{B}}

\newcommand{\psiS}{\psi_S}
\newcommand{\psihS}{\widehat{\psi}_S}
\newcommand{\psiR}{\psi_R}
\newcommand{\psihR}{\widehat{\psi}_R}

\newcommand{\mascot}{\texttt{Mascot-SDS}\xspace}
\newcommand{\fairsyn}{\texttt{Genie}\xspace}

\algblockdefx[ForRepeatUntil]{ForRepeat}{Until}
	[1]{\textbf{For} #1 \textbf{Repeat}}
	[1]{\textbf{Until} #1}

\newtheorem{remark}{Remark}

\newcommand{\supp}{\mathsf{supp}\;}

\newcommand{\dist}{\mathit{Dist}}
\newcommand{\DM}{\mathrm{DM}}

\newcommand{\Sys}{\mathfrak S}

\newcommand{\Xs}{\mathcal S}
\newcommand{\xs}{s}
\newcommand{\Us}{\mathcal U}

\newcommand{\Xh}{\widehat{\Xs}}
\newcommand{\xh}{\widehat{\xs}}

\newcommand{\Fo}{\overline{F}}
\newcommand{\Fu}{\underline{F}}

\newcommand{\RS}{\mathit{ReachSet}}

\newcommand{\cpre}{\mathit{Cpre}}

\newcommand{\apre}{\mathit{Apre}}

\newcommand{\Cont}{\rho}

\newcommand{\W}{\mathsf{WinDom}}

\newcommand{\set}[1]{\lbrace #1 \rbrace}
\newcommand{\tup}[1]{\langle #1 \rangle}
\newcommand{\dom}{\mathsf{dom}\;}

\newcommand{\REFlem}[1]{\text{Lem.~\ref{#1}}}
\newcommand{\REFthm}[1]{\text{Thm.~\ref{#1}}}
\newcommand{\REFdef}[1]{Def.~\ref{#1}}

\newcommand{\REFsec}[1]{Sec.~\ref{#1}}
\newcommand{\REFprop}[1]{Prop.~\ref{#1}}
\newcommand{\REFfig}[1]{Fig.~\ref{#1}}

\newcommand{\REFprob}[1]{Prob.~\ref{#1}}

\newcommand{\nat}{\mathbb N}
\newcommand{\incone}[2]{[#1;#2]}
\newcommand{\inctwo}[2]{\incone{#1}{#2}}

\newcommand{\parity}{\mathit{Parity}(\mathcal{P})}
\newcommand{\paritya}{\mathit{Parity}(\widehat{\mathcal{P}})}

\newcommand{\Policy}{\Pi}
\newcommand{\twohalf}{2\sfrac{1}{2}}


%% file: abstract.tex

\begin{abstract}
We consider the problem of computing the maximal probability of satisfying an $\omega$-regular specification
for stochastic, continuous-state, nonlinear systems evolving in discrete time.
The problem reduces, after automata-theoretic constructions, to finding the maximal probability of satisfying
a parity condition on a (possibly hybrid) state space.
While characterizing the exact satisfaction probability is open, we show that a lower bound on this probability can be obtained by \textbf{(I)} computing an under-approximation of the \emph{qualitative winning region}, i.e., states from which the parity
condition can be enforced almost surely, and \textbf{(II)} computing the maximal probability of reaching this qualitative winning region.

The heart of our approach is a technique to \emph{symbolically} compute the under-approximation of the qualitative winning region in step \textbf{(I)} via a \emph{finite-state
abstraction} of the original system as a \emph{$\twohalf$-player parity game}.
Our abstraction procedure uses only the support of the probabilistic evolution; it does not use precise numerical transition probabilities.
We prove that the winning set in the abstract $\twohalf$-player game induces an under-approximation of the qualitative winning region in the original synthesis problem, along with a policy to solve it.
By combining these contributions with (a) a  symbolic fixpoint algorithm to solve $\twohalf$-player games and (b) existing techniques for reachability policy synthesis in stochastic nonlinear systems,
we get an \emph{abstraction-based algorithm} for finding a lower bound on the maximal satisfaction probability.

We have implemented the abstraction-based algorithm in \mascot, where we combined the outlined abstraction step with our
tool \fairsyn~\citep{majumdar2023flexible} that solves $\twohalf$-player parity games (through a reduction to Rabin games) more efficiently than existing algorithms.
We evaluated our implementation on the nonlinear model of a perturbed bistable switch from the literature.
We show empirically that the lower bound on the winning region computed by our approach is precise, by comparing against an over-approximation of the qualitative winning region.
Moreover,
our implementation outperforms a recently proposed tool for solving this problem by a large margin.
\end{abstract}

%% file: intro.tex

\section{Introduction}

Controlled Markov processes (CMPs) over continuous state spaces and evolving in discrete time
form a general model for temporal decision making under stochastic uncertainty.
In recent years, the problem of finding or approximating optimal policies in CMPs for specifications given in temporal logics
or automata has received a lot of attention.
While there is a steady progression towards more expressive models and properties \citep{Belta17,HS_TAC19,Coogan19,majumdar2020symbolic,Coogan20,lavaei2021automated},
a satisfactory solution that can handle \emph{nonlinear} models for general \emph{$\omega$-regular specifications} in a \emph{symbolic} way
is still open.
In this paper, we make progress toward a solution to this general problem.

For \emph{finite-state} Markov decision processes (MDP), we can find optimal policies for $\omega$-regular specifications by decomposing the problem
into two parts \citep{CourcoubetisYannakakis,BdA95,Luca98,BK08}.
\begin{inparaenum}[(I)]
\item[\textbf{(I)}]
We find the set of states that ensures the satisfaction of the specification almost surely, using graph-theoretic techniques that ignore the actual transition probabilities,
Further, for any state in this \emph{almost sure winning region}, an optimal policy for \emph{almost sure} satisfaction of the specification can be derived.
\item[\textbf{(II)}]
We find an optimal policy to \emph{reach} the almost sure winning region using linear programming or traditional dynamic programming approaches.
Combining both policies returns an optimal policy for the overall synthesis problem.
\end{inparaenum}

Unfortunately, this two-step solution approach does not carry over to  optimal policy synthesis for \emph{all} $\omega$-regular
specifications given a \emph{continuous-state} CMP.
First, we do not have characterizations of optimal policies for almost sure satisfaction in this case---such as whether randomization or memory is necessary.
Second, in contrast to finite-sate MDPs, it is possible that the almost sure winning region of a CMP is empty, even if there is a policy that satisfies the specification with positive probability \citep{majumdar2020symbolic}.

However, we show in this paper, that the same decomposition can be used to compute an \emph{under-approximation} for the optimal policy instead: that is, the resulting policy gives a lower bound on the probability of satisfying a given $\omega$-regular specification from every state.
While existing techniques \citep{SA13,TMKA17,HS_TAC19} can be used in step \textbf{(II)} to compute the reachability probability with any given precision, this paper provides a new technique to under-approximate the set of states of a CMP that almost surely satisfies a \emph{parity} specification in step \textbf{(I)} of the decomposition.
A parity specification is a canonical representation for all $\omega$-regular properties \citep{EJ91,Thomas95};
thus, our approach provides a way to under-approximate the satisfaction probability for any $\omega$-regular specification.

The main contribution of our paper is to show that the approximate solution to step \textbf{(I)} can be computed by a \emph{symbolic algorithm} over a finite state \emph{abstraction} of the underlying CMP that is constructed using only the support of the probabilistic evolution of the system. This \emph{abstraction-based} policy synthesis technique is inspired by abstraction-based controller design (ABCD) for non-stochastic systems \citep{ReissigWeberRungger_2017_FRR,tabuada09,nilsson2017augmented}.
In ABCD, a nonlinear dynamical system is abstracted into a discrete two-player game over a finite discrete state space obtained by partitioning the
continuous state space into a finite set of cells.
The resulting abstract two-player game is then used to synthesize a discrete controller which is then refined into a continuous controller for the original system.

In ABCD, the abstract two-player game models the interplay between the controller ($\p{0}$) and the dynamics ($\p{1}$) such that the resulting abstract controller (i.e., the winning strategy of $\p{0}$) can be correctly refined to the original system. This requires a very powerful $\p{1}$; in every instance of the play (corresponding to the system being in one particular abstract cell $\xh$) and for any input $u$ chosen by $\p{0}$, $\p{1}$ can adversarially choose both (a) the actual continuous state $s$ within $\xh$ to which $u$ is applied and (b) any continuous disturbance affecting the system at this state~$s$.

The key insight in our work is that the stochastic nature of the underlying CMP does not require a fully adversarial treatment of continuous disturbances by $\p{1}$ in the abstract game to allow for controller refinement. Intuitively, disturbances need to be handled in a \emph{fair} way: In the long run, all transitions with positive probability will eventually occur.
We show, that the resulting fairness assumption on the behavior of $\p{1}$ can be modeled by an additional \emph{random} player (also called $\sfrac{1}{2}$ player) resulting in a so called $\twohalf$-player game \citep{Condon92,CJH03,CHSurvey} as the abstraction.

This provides a conceptually very appealing result of our paper. Using $\twohalf$-player games as abstractions of CMPs allows to utilize the machinery of symbolic game solving, analogous to ABCD techniques for non-stochastic systems, while
capturing the intuitive differences between the problem instances by the use of a random player in the abstract game. Most interestingly, the stochastic nature of the resulting abstract game \emph{eases} the abstract synthesis problem compared to standard ABCD, in the sense that controllers for a given specification are more likely to be found.
The price to be paid, however, is that solving $\twohalf$-player games is computationally harder than solving two-player games. To mitigate this problem we utilize the recent symbolic algorithm for solving $\twohalf$-player games by \cite{banerjee2022tacas,banerjee2023fast} which has almost the same worst-case complexity as the usual algorithm for (non-stochastic) two player games. With this, our approach becomes computationally tractable and vastly outperforms recent tools using different solution approaches.

In conclusion, we obtain a \emph{symbolic algorithm} to compute an under-approximation of the almost sure winning region in a \emph{continuous-state}  CMP for all $\omega$-regular specifications. Moreover, similar to the results for finite-state MDPs, this shows that the (approximate) solution to step \textbf{(I)} does not need to handle the actual transition probabilities. They are only needed in step \textbf{(II)}, where existing techniques can be used. In addition, we yield an efficient implementation of our approach by utilizing a recent symbolic game solving algorithm for $\twohalf$-player parity games.

A preliminary version of our paper appeared in ADHS '21 \citep{majumdar2021adhs}.
In this manuscript we substantially improve the conference version; the main additional contributions are as follows:
\begin{enumerate}
\item Most importantly, we have completed the exposition of our solution to the stated synthesis problem by
\begin{enumerate}
 \item additionally \emph{addressing the quantitative aspect} of the optimal policy synthesis problem in \REFthm{thm:the problem can be broken down} and \REFthm{thm:absorbing}, while the conference version of the paper only dealt with the qualitative aspect of the problem, and
 \item including the new symbolic algorithm for solving the resulting $\twohalf$-player games adapted from \cite{banerjee2022tacas}, while the conference version utilized the exising algorithm by \cite{CJH03}.
\end{enumerate}
	\item We have \emph{added the proof} of \REFthm{thm:main}, which is the main theorem of this paper and which was omitted in the conference version.
	\item We have substantially \emph{improved the experimental evaluation} (see \REFsec{sec:case_study}) of our approach by considering benchmark examples proposed by other authors, and comparing the performance of our tool against the available tool from the literature.
	\item We have provided examples and extensive explanations of the required steps throughout the manuscript.
\end{enumerate}

\smallskip

\noindent\textbf{Related Work.} Our paper extends the recent results of \citet{majumdar2020symbolic} from B\"uchi specifications to parity specifications.
Seen through the lens of $\twohalf$-player games, the algorithm of \citet{majumdar2020symbolic} can be seen as directly solving
a B\"uchi game symbolically on a non-probabilistic abstraction, by implicitly reducing the $\twohalf$-player games
to two player games on graphs with extreme fairness assumptions \citep{Pnueli83}.
While it may be possible to present a similar ``direct'' symbolic algorithm for parity games, the details of handling fairness
symbolically get difficult.
Our exposition in this paper helps separate out the different combinatorial aspects---the representation of the abstraction
and the solution of the game on the abstraction---leading to a clean proof of correctness.

$\twohalf$-player games have been used as abstractions of probabilistic systems, both in the finite case \citep{Prism-games}
and for stochastic \emph{linear} systems \citep{Belta17}.
Our paper subsumes the result of both these cases by showing a computational procedure to abstract a general, nonlinear CMP into a finite-state $\twohalf$-player game.
Further the existing approach for stochastic linear systems \citep{Belta17} only considers specifications in the GR(1) fragment of linear temporal logic, whereas we can handle any $\omega$-regular specification.
Another difference is that the linearity assumption enables the use of symbolic algorithms based on polyhedral manipulations, when the specification is given using polytopic predicates on the state space.
Instead, our abstractions are based on gridding the state space, as in ABCD for non-stochastic nonlinear systems.

Stochastic nonlinear systems were abstracted to finite-state bounded-parameter Markov decision processes (BMDP) by \citet{Coogan20} for the purpose of controller synthesis.
By using algorithms for finding controllers on BMDPs against deterministic Rabin automata, they provide an alternative approach for approximating the optimal controller against $\omega$-regular specifications.
Their method is conceptually very different.
It explicitly computes lower and upper bounds of all involved transition probabilities for constructing the BMDP, and computes winning regions by an enumerative algorithm taking these probability bounds into account.
On the other hand, our approach shows a clean separation between step \textbf{(II)}, which requires knowing explicit
transition probabilities but can be solved by existing techniques, and step \textbf{(I)}, for which this knowledge is not needed.
This allows us to provide a conceptually simpler and computationally superior \emph{symbolic algorithm} approximately solving \textbf{(I)} via abstract $\twohalf$-player games.

In principle, one can alternatively obtain a similar $\twohalf$-player game, like us, by first reducing a CMP to a BMDP using the method of \citet{Coogan20} and then reducing the BMDP to a $2\sfrac{1}{2}$-player game using the results of \citet{weininger2019satisfiability}.
However, this would still require going through the expensive construction of the BMDP, whereas we present a direct and computationally efficient way to obtain a $\twohalf$-player game through reachable set computations for the CMP.
The computational benefit of our approach over the BMDP-based method is enormous.
\new{
Theoretically, while the BMDP-based method needs \emph{exponential} time in the number of abstract states,
our algorithm can be solved in \emph{quasi-polynomial} time in the number of abstract states
using recent advances in the symbolic fixpoint computations \citep{arnold2021quasi}.
Empirically, on a couple of benchmark examples taken from the work by \citet{Coogan20}, our implementation based on uniform gridding is up to around $150$ times faster and uses up to around $150$ times smaller amount of memory than their adaptive refinement-based implementation.
Although the two implementations are not entirely comparable (ours is in C++ and theirs is in Python), we believe that the better theoretical complexity bound of our approach will manifest itself in comparable implementations as well.
The details of the experimental comparisons can be found in \REFsec{sec:case_study}.
}

%% file: prelims.tex

\section{Stochastic Nonlinear Systems}
\label{sec:stochastic_nonlinear}

\subsection{Preliminaries}
For any set $A$, a sigma-algebra on $A$ comprises subsets of $A$ as events that includes $A$ itself and is closed under complement and countable unions.
We consider a probability space $(\Omega,\mathcal F_{\Omega},P_{\Omega})$, where $\Omega$ is the sample space,
$\mathcal F_{\Omega}$ is a sigma-algebra on $\Omega$,
and $P_{\Omega}$ is a probability measure that assigns probabilities to events.
An ($(S,\mathcal F_S)$-valued) random variable $X$ is a measurable function of the form $X:(\Omega,\mathcal F_{\Omega})\rightarrow (S,\mathcal F_S)$, where $S$ is the
codomain of $X$ and $\mathcal F_S$ is a sigma-algebra on $S$.
Any random variable $X$ induces a probability measure on its space $(S,\mathcal F_S)$ as $P(\{A\}) = P_{\Omega}\{X^{-1}(A)\}$ for any $A\in \mathcal F_S$.
We often directly discuss the probability measure on $(S,\mathcal F_S)$ without explicitly mentioning the underlying probability space $(\Omega,\mathcal F_{\Omega},P_{\Omega})$  and the function $X$ itself.

A topological space $S$ is called a Borel space if it is homeomorphic to a Borel subset of a Polish space (i.e., a separable and completely metrizable space).
Examples of a Borel space are the Euclidean spaces $\mathbb R^n$, its Borel subsets endowed with a subspace topology, as well as hybrid spaces.
Any Borel space $S$ is assumed to be endowed with a Borel sigma-algebra (i.e., the one generated by the open sets in the topology), which is
denoted by $\mathcal B(S)$. We say that a map $f : S\rightarrow Y$ is measurable whenever it is Borel measurable.

Given an alphabet $A$, we use the notation $A^*$ and $ A^\omega$ to denote respectively the set of all finite words, the set of all infinite words formed using the letters of the alphabet $A$, and use $A^\infty$ to denote the set $A^*\cup A^\omega$.
Let $X$ be a set and $R\subseteq X\times X$ be a relation.
For simplicity, let us assume that $\dom{R} \coloneqq \set{x\in X\mid \exists y\in X\;.\;(x,y)\in R} = X$.
For any given $x\in X$, we use the notation $R(x)$ to denote the set $\set{y\in X\mid (x,y)\in R}$.
We extend this notation to sets: For any given $Z\subseteq X$, we use the notation $R(Z)$ to denote the set $\cup_{z\in Z}R(z)$.
Given a set $A$, we use the notation $\dist(A)$ to denote the set of all probability distributions over $A$.
\new{For any distribution $\mu\in\dist(A)$, we write $\supp{\mu}$ to denote the support of $\mu$.}

We denote the set of nonnegative integers by $\nat \coloneqq \{0,1,2,\ldots\}$ and the set of integers in an interval by $\incone{a}{b} \coloneqq \set{ a+k \mid k\in\nat,\, k\le b-a}$.
We also use the symbols ``$\ineven$'' and ``$\inodd$'' to denote memberships in the set of even and odd integers within a given set of integers:
For example, for a given set of natural numbers $M\subseteq \mathbb{N}$, the notation $n \ineven M$ is equivalent to $n\in M\cap \set{0,2,4,\ldots}$, and the notation $n\inodd M$ is equivalent to $n\in M\cap \set{1,3,5,\ldots}$.

\subsection{Controlled Markov Processes}

A \emph{controlled Markov process (CMP)} is a tuple $\mathfrak S =\left(\mathcal S, \mathcal U, \Ker\right),$
where $\mathcal S$ is a Borel space called the \emph{state space},
$\mathcal U$ is a finite set called the \emph{input space}, and
$\Ker$ is a conditional stochastic kernel $\Ker\colon\mathcal B(\mathcal S)\times \mathcal S\times \mathcal U\to [0,1]$
with $\mathcal B (\mathcal S)$  being  the Borel sigma-algebra on the state space and $(\mathcal S, \mathcal B (\mathcal S))$ being the corresponding measurable space.
The kernel $\Ker$ assigns to any $s \in \mathcal S$ and $u\in\mathcal U$ a probability measure $\Ker(\cdot | s,u)$
on the measurable space $(\mathcal S,\mathcal B(\mathcal S))$
so that for any set $A \in \mathcal B(\mathcal S), P_{s,u}(A) = \int_A \Ker (ds|s,u)$,
where $P_{s,u}$ denotes the conditional probability $P(\cdot|s,u)$.

In general, the input space $\mathcal U$ can be any Borel space and the set of valid inputs can be state dependent.
While our results can be extended to this setting, for ease of exposition, we consider the special case where
$\mathcal U$ is a finite set and any input can be taken at any state.
\new{
This choice is naturally fulfilled by control actuators with discrete actuation domains.
Real-world examples include stepper motors to actuate the movement of robotic arms \citep{quigley2011low},
nonlinear batch production with reactors and discrete actuation \citep{rawlings2017model}, airfoil roll control with plasma actuators in aerospace industry \citep{wei2013bang},
and on/off control for building temperature regulation \citep{MESG13,SA14TCL}.
When $\mathcal{U}$ is continuous/infinite, we discretize it to obtain a finite input (sub-)space.
The consequences of input-space discretization will be discussed later in Rem.~\ref{rem:on continuous inputs}.
}

The evolution of a CMP is as follows. For $k\in\nat$, let $X^k$ denote the state at the $k$-th time step and
$A^k$ the input chosen at that time. If $X^k = s \in \mathcal S$ and $A^k = u \in \mathcal U$, then
the system moves to the next state $X^{k+1}$, according
to the probability distribution $P_{s,u}$.
Once the transition into the next state has occurred, a new input is chosen, and the process is repeated.

Given a CMP $\mathfrak S$, a \emph{finite path} of length $n+1$ is a sequence
\begin{equation*}
w^n = (s^0,s^1,\ldots,s^n),\quad n\in\nat,
\end{equation*}
where $s^i\in\mathcal S$ are state coordinates of the path.
The space of all paths of length $n+1$ is denoted $\mathcal S^{n+1}$.
An \emph{infinite path} of the CMP $\mathfrak S$ is the sequence
$w = (s^0,s^1,s^2,\ldots),$
where $s^i\in\mathcal S$ for all $i\in\mathbb N$.
The space of all infinite paths is denoted by $\mathcal S^\omega$.
The spaces $\mathcal S^{n+1}$ and $\mathcal S^\omega$ are endowed with their respective product topologies and are Borel spaces.

A \emph{stationary control policy} is a universally measurable function $\rho:\mathcal S\rightarrow\mathcal U$ such that at any
time step $n\in\mathbb N$, the input $u^n$ is taken to be $\rho(s^n)\in\mathcal U$. As we only deal with stationary control policies in this paper, we simply refer to them as \emph{policies} for short.
We denote the class of all such policies by $\Policy$.
The function $\rho$ is also called \emph{state feedback controller} in control theory.

For a CMP $\mathfrak S$, any policy $\rho\in\Policy$ together with an initial probability measure $\alpha:\mathcal{B}(\mathcal S)\rightarrow[0,1]$ on the states
induces a unique probability measure on the canonical sample space of paths \citep{hll1996}, denoted by $P_\alpha^\rho$ with the expectation $\mathbb E_\alpha^\rho$.
In the case when the initial probability measure is supported on a single state $s\in\Xs$,
i.e., $\alpha({s}) = 1$, we write $P_s^{\rho}$ and $\mathbb E_s^{\rho}$ in place of $P_\alpha^\rho$ and $\mathbb E_\alpha^\rho$, respectively.
We denote the set of probability measures on $(\mathcal S,\mathcal B(\mathcal S))$ by $\mathfrak D$.

Given any $\omega$-regular specification $\varphi$ defined using a set of predicates over the state space $\Xs$ of $\Sys$, we use the notation $\Sys\models \varphi$ to denote the set of all paths of $\Sys$ which satisfy $\varphi$.
Thus, $P_\alpha^\rho(\Sys\models \varphi)$ denotes the probability of satisfaction of $\varphi$ by $\Sys$ under the effect of the control policy $\rho$, when the initial probability measure is given by $\alpha$.
Often we will use Linear Temporal Logic (LTL) notation to express $\omega$-regular properties.
The syntax and semantics of LTL can be found in standard literature \citep{BK08}.

A \emph{stochastic dynamical system} $\Sigma$ is described by a state evolution
	\begin{equation}
	s^{k+1}=f(s^k,u^k,\varsigma^k), \ \ \ k\in\mathbb{N},
	\label{eq:state_evolution}
	\end{equation}
where $s^k\in \mathcal S$ and $u^k \in \mathcal{U}$ are states and inputs for each $k\in\nat$,
and $(\varsigma^0,\varsigma^1,\ldots)$ is assumed to be a sequence of independent and identically distributed (i.i.d.) random variables
representing a stochastic disturbance.
The map $f$ gives the next state as a function of current state, current input, and the disturbance.
One can construct a CMP over states $\mathcal{S}$ and inputs $\mathcal{U}$ from \eqref{eq:state_evolution} by noticing that for any given state
$s^k$ and input $u^k$ at time $k$, the next state is a random variable defined as a function of $\varsigma^k$.
Thus, $\Ker(\cdot|s^k,u^k)$ is exactly the distribution of the random variable $f(s^k,u^k,\varsigma^k)$ and can be
computed based on the distribution of $\varsigma^k$ and the map $f$ itself \citep{k2002}.

\subsection{Parity Specifications}

Let $\mathfrak S =\left(\mathcal S, \mathcal U, \Ker\right)$ be a CMP and suppose
$\mathcal{P}=\tup{B_1,\ldots,B_\ell}$ is a partition of $\mathcal S$ with measurable sets $B_1$, $\ldots$, $B_\ell$;
that is, $B_i \cap B_j = \emptyset$ for $i\neq j$ and $\cup_{i=1}^\ell B_i = \mathcal S$.
We allow some $B_i$'s to be empty.
For each $B_i$, we call the integer $i$ its \emph{priority}.

Intuitively, an infinite path $w\in\mathcal S^\omega$ satisfies the \emph{parity specification} with respect to $\mathcal{P}$
if the highest subset $B_i$ visited infinitely often has even priority.
This specification is formalized in Linear Temporal Logic (LTL) notation~\citep{BK08} using the following formula:
\begin{equation}\label{eq:definition of parity}
	\parity \coloneqq \bigwedge_{i\inodd\incone{1}{\ell}}\left(\square\lozenge B_{i} \implies \bigvee_{j\ineven \incone{i+1}{\ell}} \square\lozenge B_{j}\right),
\end{equation}
which requires that infinitely many visits to an odd priority subset ($\square\lozenge B_{i}$) must
imply infinitely many visits to a \emph{higher} even priority subset ($\square\lozenge B_{j}$).
We indicate the set of all infinite paths $w\in \mathcal S^\omega$ of a CMP $\mathfrak{S}$ that satisfy the property $\parity$ by $\Sys \models \parity$.
The proof of measurability of the event $\Sys \models \parity$ goes back to the work by \citet{Vardi85}
that provides the proof for probabilistic finite state programs.
The proof for a CMP follows similar principles, using the observation that
$\Sys\models\parity$ can be written as a Boolean combination of events $\Sys\models\square\lozenge A$, where $A$ is a measurable set,
and $\square \lozenge A$ is a canonical $G_\delta$ set in the Borel hierarchy.

It is well-known that every $\omega$-regular specification whose propositions range over measurable subsets of the state space of a CMP
can be modeled as a deterministic parity automaton \cite[Thm.~1.19]{gradel2002automata}.
By taking a synchronized product of this parity automaton with the CMP, we can obtain a product CMP and a parity specification over the product state space
such that every path satisfying the parity specification also satisfies the original $\omega$-regular specification and vice versa.
Moreover, a stationary policy for the parity objective gives a (possibly history-dependent) policy for the original specification.
Thus, without loss of generality, we assume that an $\omega$-regular objective is already
given as a parity condition using a partition of the state space of the system.

\section{Problem Definition}
\label{sec:problem}

We are interested in computing the maximal probability that a given parity specification can be satisfied by paths of a CMP $\mathfrak S$ starting from
a particular state $s\in\mathcal S$ under stationary policies.
Given a control policy $\rho\in\Policy $ and an initial state $s\in\mathcal S$, we define the satisfaction probability and the supremum satisfaction probability as
\begin{align}
f(s,\rho) &\coloneqq P_s^{\rho}(\mathfrak S\models \parity)~\text{and}\label{eq:sat_prob}\\
f^\ast(s) &\coloneqq \sup_{\rho\in\Policy} P_s^{\rho}(\mathfrak S\models \parity),\label{eq:optimal_prob}
\end{align}
respectively.
An \emph{optimal control policy} for the parity condition is a policy $\rho^\ast$ such that $f^\ast(s) = f(s,\rho^\ast)$ for all $s\in\mathcal S$.
Note that an optimal policy need not exist, since the supremum may not be achieved by any policy. Our goal is to study the following \emph{optimal policy synthesis} problem.

\begin{problem}[Optimal Policy Synthesis]
	\label{prob:non-trivial-parity}
	Given $\Sys$ and a parity specification  $\parity$, find an optimal control policy $\rho^*$, if it exists, together with $f^\ast(s)$ such that $P_s^{\rho^*}(\mathfrak S\models \parity) = f^\ast(s)$.
\end{problem}

While the satisfaction probability \eqref{eq:sat_prob} and the supremum satisfaction probability \eqref{eq:optimal_prob} are both well-defined,
we are not aware of any work characterizing necessary or sufficient conditions for existence of optimal control policies on continuous-space CMPs for parity specifications.
Additionally, we restrict attention to \emph{stationary} policies.
While it is possible to define more general classes of policies, that depend on the entire history and use randomization over $\mathcal U$, we are again
unaware of any work that characterizes the class of policies that are sufficient for optimal control of CMPs for parity specifications.
For finite-state systems, stationary policies are sufficient and we restrict attention to them.

Since we cannot prove existence or computability of optimal policies, in this paper, we focus on providing an
approximation procedure to compute a possibly sub-optimal control policy and guaranteed lower bounds on the optimal satisfaction probability.
Our procedure relies on first approximating almost sure winning regions (i.e., where the specification can be satisfied with probability one),
and then solving a reachability problem, as formalized next.

\begin{definition}[Almost sure winning region]
Given a CMP $\mathfrak S$, a policy $\rho$, and a parity specification $\parity$,
the state $s\in\mathcal S$ satisfies the specification \emph{almost surely} if $f(s,\rho)=1$.
The almost sure \emph{winning region}---or simply the winning region---of the policy $\rho$ is defined as
\begin{equation}
	  \W(\Sys,\rho) \coloneqq \set{ s\in\mathcal S \mid f(s,\rho) = 1}.
\end{equation}
We also define the \emph{maximal almost sure winning region}---or simply the maximal winning region---as
\begin{equation}
	 \W^\ast(\Sys) \coloneqq \set{ s\in\mathcal S \mid f^\ast(s) = 1}.
\end{equation}
\end{definition}
Note that $\W(\Sys,\rho)\subseteq\W^\ast(\Sys)$ for any control policy $\rho\in\Policy$.
It is clear by definition of the winning region that for any policy $\rho$, the satisfaction probability $P_s^{\rho}(\mathfrak S\models \parity)$ is equal to $1$ for any $s$ in the winning region $W:=\W(\Sys,\rho)$.
The next theorem states that this satisfaction probability is lower bounded by the probability of reaching the winning region $W$ for any $s\not\in W$. We denote such a reachability by $(\mathfrak S\models \lozenge W) :=\set{w = (s^0,s^1,s^2,\ldots) \mid \exists n\in\nat \,.\, s^n\in W}$.

\begin{theorem}
\label{thm:the problem can be broken down}
For any control policy $\rho\in\Policy$ on CMP $\Sys$,
and $W \coloneqq \W(\Sys,\rho)$, we have
\begin{equation}
\label{eq:decom}
\begin{array}{l l}
P_s^{\rho}(\mathfrak S\models \parity )= 1 & \text{if } s\in W \text{ and}\\
P_s^{\rho}(\mathfrak S\models \parity )\ge P_s^{\rho}(\mathfrak S\models \lozenge W) & \text{if } s\notin W.
\end{array}
\end{equation}
\end{theorem}

The proof can be found in the appendix.
The inequality in the second part of \eqref{eq:decom} is because the
$\parity$ specification may be satisfied with positive probability even though the path always stays outside of $W$.
When the state space is finite (i.e., for finite Markov decision processes), equality holds \citep{BK08}.
However, equality need not hold for general CMPs: \citet{majumdar2020symbolic} showed an example
where the maximal almost sure winning region is empty even though a
B\"uchi specification is satisfied with positive probability.

The next theorem establishes that for any policy $\rho$, the winning region is an absorbing set, i.e., paths starting from this set will stay in the set almost surely.
\begin{theorem}
\label{thm:absorbing}
For any control policy $\rho$, The set $W = \W(\Sys,\rho)$ is an absorbing set, i.e., $\Ker(W|s,\rho(s) )= 1$ for all $s\in W$. This implies $P_s^{\rho}(\Sys\models\lozenge \mathcal S\backslash W) = 0$ for all $s\in W$.
\end{theorem}
The proof of this theorem utilizes the fact that $\parity$ is a long-run property and its satisfaction is independent of the prefix of a path. The proof is provided in the appendix.

Notice that \REFthm{thm:the problem can be broken down} and \REFthm{thm:absorbing} are stated for any fixed control policy $\rho$, but these theorems enable us to decompose the maximization of
$P_s^{\rho}(\mathfrak S\models \parity)$ with respect to $\rho$ into two sub-problems.
First, find a policy that gives the largest winning region $W$ and employ that policy when the current state is in $W$.
Then, find a policy that maximizes the reachability probability $P_s^{\rho}(\mathfrak S\models \lozenge W)$ and employ that policy as long as the current state is not in $W$.

Computation of the reachability probability has been studied extensively in the literature for both
infinite horizon \citep{TA12,TMKA17,HS_TAC19,HNS20_MultiObjective} and
finite horizon \citep{SA13,SAH12,LAB15,Kariotoglou17,LO17,SAM17,Vinod18,LAVAE19,Jagtap2019}
using different abstract models and computational methods.
Together with an algorithm that underapproximates the region of almost sure satisfaction,
these approaches can be used to provide a lower bound on the probability of satisfaction of the parity condition.
In the rest of the paper, we focus on the following problem (the first half of \eqref{eq:decom}).

\begin{problem}[Approximate Maximal Winning Region]
\label{prob:policy}
Given $\Sys$ and a parity specification $\parity$, find a (sub-optimal) control policy $\rho\in \Policy$, its winning region $\W(\Sys,\rho)\neq\emptyset$, and a bound on the absolute volume of the set difference $\W^\ast(\Sys)\backslash\W(\Sys,\rho)$, which we call the \emph{approximation error}.
\end{problem}

In \REFsec{synthesis_game}-\ref{sec:soundness}, we provide a solution for \REFprob{prob:policy} via the paradigm of abstraction-based controller design.
Not surprisingly, we get a tighter (i.e., larger) approximation of $\W^\ast(\Sys)$ if we use a finer discretization of the state space during the abstraction step.
We also provide an over-approximation of $\W^\ast(\Sys)$, and show closeness of the under- and over-approximation of $\W^\ast(\Sys)$ in the numerical example provided in \REFsec{sec:case_study}.

%% file: two-and-half-game.tex

\section{Abstraction-Based Policy Synthesis}
\label{synthesis_game}

The main result of our paper is a solution to Prob.~\ref{prob:policy} via a \emph{symbolic algorithm} over abstract $\twohalf$-player games in the spirit of abstraction-based controller design (ABCD).
ABCD is typically used to compute temporal-logic controllers for \emph{non-stochastic} nonlinear dynamical systems
\citep{ReissigWeberRungger_2017_FRR,tabuada09,nilsson2017augmented} in two steps.
First, the system is abstracted into an abstract finite-state two-player game. This game is then used to synthesize a discrete controller which is then refined into a continuous controller for the original system.
In standard ABCD techniques, the abstract game has two players: $\p{0}$ simulating the controller and choosing the next control input $u$ based on the currently observed abstract state $\xh$, and $\p{1}$ simulating the adversarial effect of (a) choosing any continuous state $s$ in $\xh$ to which $u$ is applied and (b) choosing any continuous disturbance $d$ that effects the resulting transition.

The key insight in our abstraction step is that the stochastic nature of the underlying CMP allows choosing disturbances in a \emph{fair} random way instead of purely adversarially.
We model this by introducing an additional \emph{random} player (also called $\frac{1}{2}$ player) resulting in a so called $\twohalf$-player game~\citep{Condon92,CJH03,CHSurvey}. In the resulting abstract game, only the effect of the discretization is handled by $\p{1}$ in an adversarial manner. The random player picks the applied disturbance uniformly at random.

After introducing necessary preliminaries on $\twohalf$-player games in \REFsec{sec:twohalf:prelim}, we show how a CMP can be abstracted into a
$\twohalf$-player game in \REFsec{sec:abstraction}.
We then recall the basics of symbolic fixpoint computations in \REFsec{sec:fixpt:prelim}.
In \REFsec{sec:abstractsynt}, we show a symbolic procedure to find winning regions in $\twohalf$-player games for parity specifications.
Finally, we state in \REFsec{sec:refine} how an almost-sure winning strategy in the abstract $\twohalf$-player game is refined,
and that the resulting control policy is almost sure winning for the original CMP and its associated parity specification.
This establishes soundness of our ABCD technique to solve Problem~\ref{prob:policy}.

\subsection{Preliminaries: $2\sfrac{1}{2}$-Player Parity Games}\label{sec:twohalf:prelim}
\noindent
\textbf{The game graph.}
A $2\sfrac{1}{2}$-player \emph{game graph} is a tuple $\game = \tup{V,E,\tup{V_0,V_1,V_r}}$, where $V$ is a finite set of vertices, $E$ is a set of directed edges $E\subseteq V\times V$, and the sets $V_0,V_1,V_r$ form a partition of the set $V$.
A \emph{Markov chain} is a particular type of $\twohalf$-player game graph with $V=V_r$ (i.e.\ $V_0=V_1=\emptyset$).

A $2\sfrac{1}{2}$-player \emph{parity game} is a pair $\tup{\game,\mathcal{P}}$, where $\game$ is a $2\sfrac{1}{2}$-player game graph, and $\mathcal{P} = \tup{B_1,\ldots,B_\ell}$ is a tuple of $\ell$ disjoint subsets of $V$, some of which can possibly be empty.
The tuple $\mathcal{P}$ induces the parity specification $\parity$ over the set of vertices $V$ in the natural way.
In order to ensure that $\parity$ is well defined, we impose the restriction that every infinite run must have infinitely many occurrences of vertices from at least one of the nonempty sets in $\mathcal{P}$.
In other words, we require that every set of vertices $U\subseteq V$ for which there is no $i\in [1;\ell]$ with $U\cap B_i\neq \emptyset$ must be ``transient'' vertices.

\smallskip
\noindent
\textbf{The players and their strategies.}
We assume that there are two players $\p{0}$ and $\p{1}$, who are playing a game by moving a token along the edges of the game graph $\game$.
In every step, if the token is located in a vertex in $V_0$ or $V_1$, $\p{0}$ or $\p{1}$ respectively moves the token to one of the successors according to the edge relation $E$.
On the other hand, if the token is located in a vertex $v\in V_r$, then in the next step the token moves to a vertex $v'$ that is chosen \emph{uniformly at random} from the set $E(v)$.
Strategies of $\p{0}$ and $\p{1}$ are respectively the functions $\pi_0\colon V^*V_0\to \dist(V)$ and $\pi_1\colon V^*V_1\to \dist(V)$ such that for all $w\in V^*$, $v_0\in V_0$ and  $v_1\in V_1$, we have $\supp\pi_0(wv_0)\subseteq E(v_0)$ and $\supp\pi_1(wv_1)\subseteq E(v_1)$.
We use the notation $\Pi_0$ and $\Pi_1$ to denote the set of all strategies of $\p{0}$ and $\p{1}$ respectively.
A strategy $\pi_i$ of $\p{i}$, for $i\in \set{0,1}$, is \emph{deterministic memoryless}
if for every $w_1,w_2\in V^*$ and for every $v\in V_i$, $\pi_i(w_1v) = \pi_i(w_2v)$ holds; we simply write $\pi_i(v)$ in this case.
We use the notation $\Pi^{\mathrm{DM}}_i$ to denote the set of all deterministic memoryless strategies of $\p{i}$.
Observe that $\Pi^{\mathrm{DM}}_i\subseteq \Pi_i$.

\smallskip
\noindent
\textbf{Runs and winning conditions.}
An infinite (finite) run of the game graph $\game$, compatible with the strategies $\pi_0\in \Pi_0$ and $\pi_1\in \Pi_1$, is an infinite (a finite) sequence of vertices $r=v^0v^1v^2\ldots$ ($r=v^0\ldots v^n$ for some $n\in \mathbb{N}$) such that for every $k\in\nat$,
(a) $v^k\in V_0$ implies $v^{k+1}\in \supp\pi_0(v^0\ldots v^k)$,
(b) $v^k\in V_1$ implies $v^{k+1}\in \supp\pi_1(v^0\ldots v^k)$, and
(c) $v^k\in V_r$ implies $v^{k+1}\in E(v^k)$.
Given an initial vertex $v^0$ and a fixed pair of strategies $\pi_0\in \Pi_0$ and $\pi_1\in \Pi_1$, we obtain a probability distribution over the set of infinite runs of the system.
For a measurable set of runs $R\subseteq V^\omega$, we use the notation $P_{v^0}^{\pi_0,\pi_1}(R)$ to denote the probability of obtaining the set of runs $R$ when the initial vertex is $v^0$ and the strategies of $\p{0}$ and $\p{1}$ are fixed to respectively $\pi_0$ and $\pi_1$.
For an $\omega$-regular specification $\varphi$, defined using a predicate over the set of vertices of $\game$, we write $(\game\models \varphi
)$ to denote the set of all infinite runs for all possible strategies of $\p{0}$ and $\p{1}$ that satisfy $\varphi$.
For example, $(\game\models \parity)$ denotes the set of all infinite runs for all possible strategies of $\p{0}$ and $\p{1}$ that satisfy the parity condition $\parity$.
We say that $\p{0}$ wins $\parity$ almost surely from a vertex $v\in V$ (or $v$ is almost sure winning for $\p{0}$) if $\p{0}$ has a strategy $\pi_0\in \Pi_0$ such that for all $\pi_1\in \Pi_1$ we have
$
P_v^{\pi_0,\pi_1}(\game\models \parity)=1
$.
We collect all vertices for which this is true in the almost sure winning region $\mathcal{W}(\game\models \parity)$.

\subsection{Abstraction: CMPs to $2\sfrac{1}{2}$-Player Games}\label{sec:abstraction}

Given a CMP $\mathfrak{S}=(\mathcal{S},\mathcal{U},\Ker)$ and a parity specification $\mathit{Parity}(\mathcal{P})$ for
a partition $\mathcal{P}$ of the state space $\mathcal{S}$ we construct an abstract $2\sfrac{1}{2}$-player game.

\smallskip
\noindent\textbf{State-space abstraction.}
We introduce a finite partition $\Xh\coloneqq \set{\xh_i}_{i\in I}$ such that $\cup_{i\in I} \xh_i = \mathcal{S} $, $\xh_i\neq\emptyset$ and $\xh_i\cap \xh_j=\emptyset$ for every $\xh_i,\xh_j\in \Xh$ with $i\neq j$.
Furthermore, we assume that the partition $\Xh$ is consistent with the given priorities $\mathcal{P}$, i.e., for every partition element $\xh\in \Xh$, and for every $x,y\in \xh$, $x$ and $y$ belong to the same partition element in $\mathcal{P}$ (i.e., $x$ and $y$ are assigned the same priority).
We call the set $\Xh$ the \emph{abstract state space} and each element $\xh\in \Xh$ an \emph{abstract state}.

We introduce the abstraction function $Q\colon \Xs\to \Xh$ as a mapping from the continuous to the abstract states: For every $s\in \Xs$, $Q\colon s\mapsto \xh$ such that $s\in \xh$.
We define the concretization function as the inverse of the abstraction function: $Q^{-1}\colon \Xh\to 2^\Xs$, $Q^{-1}\colon \xh\mapsto \set{s\in \Xs \mid s\in \xh}$.
We generalize the use of $Q$ and $Q^{-1}$ to sets of states:
For every $U\subseteq \Xs$, $Q(U) = \bigcup_{s\in U} Q(s)$, and for every $\widehat{U}\subseteq \Xh$, $Q^{-1}(\widehat{U}) = \bigcup_{\xh\in \widehat{U}} Q^{-1}(\xh)$.

\smallskip
\noindent\textbf{Transition abstraction.}
We also introduce \new{over- and under-approximations of the probabilistic transitions of the CMP $\Sys$ denoted by the non-deterministic abstract transition functions $\Fo\colon \Xh\times \mathcal{U} \to 2^{\Xh}$ and $\Fu\colon \Xh\times \mathcal{U} \to 2^{\Xh}$. These two approximations are required to satisfy} the following properties:
\begin{subequations}\label{equ:FuFo}
 \begin{align}
	\Fo(\xh,u) &\supseteq \set{\xh'\in\Xh \mid \exists s\in \xh\;.\;\Ker(\xh'\mid s,u) > 0},\label{eq:def fo}\\
	\Fu(\xh,u) &\subseteq \set{\xh'\in\Xh\mid \exists \varepsilon>0\;.\; \forall s\in \xh\;.\;\Ker(\xh'\mid s,u) \geq \varepsilon}.\label{eq:def fu}
\end{align}
\end{subequations}

Intuitively, given an abstract state $\xh$ and an input $u$, the set $\Fo$ over-approximates the set of abstract states reachable by probabilistic transitions from $\xh$ on input $u$. On the other hand,  $\Fu$ under-approximates those abstract states which can be reached by \emph{every} state in $\xh$
with probability bounded away from zero.

To understand the need for both $\Fo$ and $\Fu$ and the way they are constructed, consider the following example.

\begin{example}\label{ex:FoFu}
Consider the two CMPs, $\Sys_A$ and $\Sys_B$:

 \begin{center}
\begin{tikzpicture}[node distance=1cm]
	\pgfdeclarelayer{bg}    
	\pgfsetlayers{bg,main}  
	\node[] at (-1,0) {$\Sys_A:$};

	\tikzstyle{every node}=[circle,inner sep=2pt]
	\newcommand{\offset}{0.4cm}
	\colorlet{lightgrey}{gray!30}
	\node[state,minimum size=10pt]		(A)	at		(0,0)		{$s_1$};
	\node[state,minimum size=10pt]		(B)	[right of=A]		{$s_2$};
	\node[state,minimum size=10pt]		(C)	[right of=B]		{$s_3$};

	\node[right of= A, xshift=-0.5cm, yshift=-0.7cm] (sh1) {$\xh_1$};
	\node[below of= C, yshift=0.3cm] (sh2) {$\xh_2$};

	\begin{pgfonlayer}{bg}    
        \draw[draw, dashed, fill=lightgrey,rounded corners]		($(A)-(\offset,\offset)$)		rectangle		($(B)+(\offset,\offset)$);
        \draw[draw, dashed, fill=lightgrey,rounded corners]		($(C)-(\offset,\offset)$)		rectangle		($(C)+(\offset,\offset)$);
    \end{pgfonlayer}

	\path[->]
			(A)	edge[loop above]		()
					edge[bend left]							(C)
			(B)	edge[loop below]		()
					edge[bend right]							(C);


	\node[right of= C, xshift=0.5cm] {$\Sys_B:$};
	\node[state,minimum size=10pt]		(D)	[right of= C, xshift=1.5cm]		{$s_1$};
	\node[state,minimum size=10pt]		(E)	[right of=D]		{$s_2$};
	\node[state,minimum size=10pt]		(F)	[right of=E]		{$s_3$};
	\node[right of= D, xshift=-0.5cm, yshift=-0.7cm] (sh3) {$\xh_1$};
	\node[below of= F, yshift=0.3cm] (sh4) {$\xh_2$};

	\begin{pgfonlayer}{bg}    
        \draw[draw, dashed, fill=lightgrey,rounded corners]		($(D)-(\offset,\offset)$)		rectangle		($(E)+(\offset,\offset)$);
        \draw[draw, dashed, fill=lightgrey,rounded corners]		($(F)-(\offset,\offset)$)		rectangle		($(F)+(\offset,\offset)$);
    \end{pgfonlayer}

	\path[->]
			(D)	edge[loop above]		()
			(E)	edge[loop below]		()
					edge[bend right]							(F);
\end{tikzpicture}
\end{center}

The circles are concrete states $s_i$, the dashed boxes denote abstract states $\xh_i$, and the edges denote transitions with
positive probability between concrete states $s_i$.
Consider the left abstract state $\xh_1$. Here, the adversary decides which concrete state (i.e., $s_1$ or $s_2$) the game is in.
In both $\Sys_A$ and $\Sys_B$, $\Fo(\xh_1)$ contains both $\xh_1$ and $\xh_2$, because both of them are reachable from $\xh_1$.
On the other hand, in $\Sys_A$, $\Fu(\xh_1)$ contains both $\xh_1$ and $\xh_2$, while in $\Sys_B$, $\Fu(\xh_1)$ contains only $\xh_1$.
An adversary that plays according to $\Fo$ is too strong: it can keep playing the self loop in $s_2$,
while the stochastic nature of the CMP ensures that eventually $s_2$ will transition to $s_3$.
In order to follow the probabilistic semantics, we must ensure the adversary
picks a distribution whose support contains both abstract states.

In $\Sys_B$, the probabilistic behavior of the two concrete states $s_1$ and $s_2$ are very different:
$s_1$ stays in $\xh_1$ with probability one and $s_2$
stays in $\xh_1$ or moves probabilistically to $\xh_2$.
To ensure correct behavior, we look at possible supports of distributions induced by the dynamics: these
are the possible subsets of abstract states between $\Fu$ and $\Fo$.
Here, the game either stays in $\xh_1$ or (eventually) moves to $\xh_2$ and, in our reduction, we force the adversary
to commit to one of the two options.
\end{example}

\new{
The choices of $\Fo(\cdot)$ and $\Fu(\cdot)$ that satisfy \eqref{equ:FuFo} are not unique. The theoretical results provided in this paper are valid for any choices of $\Fo(\cdot)$ and $\Fu(\cdot)$ that satisfy \eqref{equ:FuFo}.
In practice, the tighter the implementations of $\Fo(\cdot)$ and $\Fu(\cdot)$ are, the smaller will be the approximation error for the optimal winning region $\W^\ast(\Sys)$.
The approximation error will be the largest for the extreme case of $\Fo\equiv\Xh$ and $\Fu\equiv\emptyset$, which will produce a sound but trivial control policy with empty winning region (assuming that the set of odd priority vertices is non-empty).
}

	The parameter $\varepsilon$ states that there is a uniform lower bound on transition probabilities for all states in an abstract state.
	This ensures that, provided $\xh$ is visited infinitely often and $u$ is applied infinitely often from $\xh$, then
	$\xh'$ will be reached almost surely from $\xh$.
	In the absence of a uniform lower bound, this property need not hold for infinite state systems; for example, if the probability converges to zero,
	the probability of escaping $\xh$ can be strictly less than one.

\smallskip
\new{
\noindent\textbf{Efficient computation of $\Fo$ and $\Fu$ for disturbance-affine systems.}
}
\new{We show how to efficiently compute the over- and under-approximations $\Fo$ and $\Fu$ for the important}
subclass of stochastic nonlinear systems with \emph{affine} disturbances
\begin{align*}
	\xs^{k+1} = f(\xs^k,u^k) + \varsigma^k,	\quad k\in \mathbb{N},
\end{align*}
where $\varsigma^0, \varsigma^1, \ldots$ are independent and identically distributed random variables from a distribution with a bounded support $D$,
and we assume we are only interested in a compact region $\Xh'$ of the state space.
In this case, for any abstract state $\xh$ and any $u\in U$, one can compute an approximation $\RS(\xh, u)$ with
$\RS(\xh,u) \supseteq \set{\xs' \in \Xs \mid \exists \xs\in \xh\;.\; f(\xs,u)=\xs'}$
using standard techniques \citep{ReissigWeberRungger_2017_FRR,coogan2015efficient,althoff2013reachability}.
Define $S_1,S_2\colon 2^{\Xh}\times \Us\to 2^{\Xs}$ such that
\begin{align*}
	& S_1\colon (\xh,u)\mapsto {D}\oplus \RS(\xh,u),\\
	& S_2\colon (\xh,u)\mapsto {D}\ominus (-\RS(\xh,u)),
\end{align*}
where the minus sign ($-\RS(\xh,u)$) is applied to each individual element of $\RS(\xh,u)$ and $\oplus$ and $\ominus$
are Minkowski sum and difference, respectively.
Given $S_1$ and $S_2$, the
\new{
tightest (with respect to the specific implementation of $\RS(\cdot)$) possible implementation of the
}
 functions $\Fo(\cdot,\cdot)$ and $\Fu(\cdot,\cdot)$ are given as \citep[Thm.~6.1]{majumdar2020symbolic}:
(1) $\xh'\in \Fo(\xh,u)$  iff either $S_1(\xh,u) \subseteq \Xs'$ and $\xh'\cap S_1(\xh,u)\neq \emptyset$ or
$S_1(\xh,u) \not\subseteq \Xs'$ and $\xh'$ is a special sink state; and
(2) $\xh'\in\Fu(\xh,u)$ iff either $\lambda(\xh'\cap S_2(\xh, u)) > 0$ or
$\lambda(S_2(\xh,u) \setminus \Xs') > 0$ and $\xh'$ is a special sink state, and
where $\lambda(\cdot)$ denotes the Lebesgue measure (generalized volume) of a set.

\begin{remark}
 We remark that standard algorithms \citep{ReissigWeberRungger_2017_FRR, coogan2015efficient} for computing $\RS(\cdot,\cdot)$,
including the one that is used in our implementation, have the following \emph{monotonicity} property:
For every $\xh,\xh'\in \Xh$ and for every $u\in \Us$, if $\xh\subseteq \xh'$ then
\begin{equation*}
	 \RS(\xh,u) \subseteq \RS(\xh',u).
\end{equation*}
Monotonicity of $\RS(\cdot,\cdot)$ implies $\Fo(\cdot,\cdot)$ and $\Fu(\cdot,\cdot)$ are monotone:
For every $\xh,\xh'\in \Xh$ and for every $u\in \Us$, if $\xh\subseteq \xh'$ then
\begin{align}\label{eq:monotonic fo and fu}
	\Fo(\xh,u) \subseteq \Fo(\xh',u) &&\text{ and}
		&&\Fu(\xh,u) \supseteq \Fu(\xh',u).
\end{align}
This property ensures that  we get a better (i.e., larger) approximation of the optimal almost sure winning domain $\W^\ast(\Sys)$ in problem~\ref{prob:policy} if we use a finer discretization of the state space during the abstraction step. This observation can be empirically confirmed from the experiments in \REFsec{sec:case_study}.
\new{
On the other hand, if $\RS(\cdot,\cdot)$ is not monotone, then we cannot hope to experience a similar improvement in the quality of the approximation for finer discretizations.
Needless to say, our theoretical claims still remain valid.
}
\end{remark}

\smallskip
\new{
\noindent\textbf{Computation of $\Fo$ and $\Fu$ for the general class of systems.}
Now consider the general class of CMPs of the form
\begin{align*}
	\xs^{k+1} = f(\xs^k,u^k,\varsigma^k),	\quad k\in \mathbb{N},
\end{align*}
where $f:\Xs\times\Us\times Z\rightarrow\Xs$ is a nonlinear function and $\varsigma^0, \varsigma^1, \ldots$ are independent and identically distributed random variables from a distribution $T_{\zeta}(\cdot)$ with a bounded support $D\subset Z$. Denote the open ball centered at $\zeta\in Z$ and radius $r> 0$ by  $\mathsf B_r(\zeta) := \{\varsigma\in Z\mid \|\varsigma-\zeta\|<r\}$.
Then,
$\Fo$ and $\Fu$ can be constructed as:
\begin{align*}
	\xh'\in \Fo(\xh,u)  &\Leftarrow \exists \xs\in\xh \,\cdot\,\exists r>0 \,\cdot\, \exists \zeta\in D   \,\cdot\, \forall \zeta'\in \mathsf B_r(\zeta) \,\cdot\, \zeta'\in D \,\wedge\,  f(\xs,u,\zeta')\in \xh',\\
	\xh'\in\Fu(\xh,u)	 &\Leftarrow \exists r>0 \,\cdot\, \forall \xs\in\xh \,\cdot\,\exists \zeta\in D \,\cdot\, \forall \zeta'\in \mathsf B_r(\zeta) \,\cdot\, \zeta'\in D \,\wedge\,  f(\xs,u,\zeta')\in \xh'.
\end{align*}
 When $D$ is semi-algebraic and $f$ is a polynomial function of its arguments, the above formulas belong to the theory of reals with addition and multiplication,
which is decidable but has a high computational complexity.
While there are some implementations of the theory, their practical performance in computing abstractions remains to be evaluated.

Another approach to construct $\Fo$ and $\Fu$ is to use the following over- and under-approximations:
\begin{subequations}
\label{equ:FuFo2}
 \begin{align}
 \Fo(\xh,u) &\supseteq \left\{\xh'\in\Xh \,\bigg|\, \bigcup\limits_{\xs\in\xh}g(\xs,u,\xh') \neq\emptyset\right\},\label{eq:over_approx_reach}\\
	\Fo(\xh,u) &\subseteq \left\{\xh'\in\Xh \,\bigg|\, \lambda\left(\bigcap\limits_{\xs\in\xh}g(\xs,u,\xh')\right)>0\right\},\label{eq:under_approx_reach}
\end{align}
\end{subequations}
where $g(\xs,u,\xh'):=\{\zeta\in D \mid f(\xs,u,\zeta)\in \xh' \}$.

Intuitively, the condition in the right side of \eqref{eq:over_approx_reach} implies that there is at least one state $\xs\in\xh$ such that the set $g(\xs,u,\xh')\subseteq D$ is not empty. This includes all the cases that satisfy $\Ker(\xh'\mid \xs,u)>0$. Therefore, the right side of \eqref{eq:over_approx_reach} is a superset of the right side of \eqref{eq:def fo} implying that any $\Fo$ satisfying \eqref{eq:over_approx_reach} conforms to the definition in \eqref{eq:def fo}.
The condition in the right side of \eqref{eq:under_approx_reach} implies that there is a subset of $D$ with positive measure (i.e., the set $\cap_{\xs\in\xh}g(\xs,u,\xh')$) such that from every state $\xs\in \xh$, the next state is in $\xh'$ under the control input $u$ and the disturbances from this set. This means $\Ker(\xh'\mid \xs,u)\ge \varepsilon$ for all $\xs\in\xh$ with $\varepsilon = T_\zeta(\cap_{\xs\in\xh}g(\xs,u,\xh'))>0$.
Therefore, the right side of \eqref{eq:under_approx_reach} is a subset of the right side of \eqref{eq:def fu} implying that any $\Fu$ satisfying \eqref{eq:under_approx_reach} conforms to the definition in \eqref{eq:def fu}.

Note that $\cup_{\xs\in\xh}g(\xs,u,\xh')$ and $\cap_{\xs\in\xh}g(\xs,u,\xh')$ can respectively be over- and under-approximated depending on the structure of the function $f$ and the shape of the sets $\xh$, $\xh'$, and $D$ (e.g., unions or intersections of hyper-rectangles, ellipsoids, polytopes, or zonotopes). This type of approximations has been studied in the (backward) reachability analysis of nonlinear systems by considering set propagations under nonlinear mappings \citep{althoff2021set}.
We also note that the computation of $\Fo$ satisfying \eqref{eq:over_approx_reach} can be done using forward reachability techniques applied to the following equivalent condition
\begin{equation*}
 \Fo(\xh,u) \supseteq \left\{\xh'\in\Xh \,\bigg|\, \bigcup\limits_{\xs\in\xh}f(\xs,u,D) \cap \xh'\neq\emptyset\right\},
\end{equation*}
where $f(\xs,u,D) := \{f(\xs,u,\zeta) \mid \zeta\in D\}$.
}

\medskip

\noindent\textbf{Abstract $2\sfrac{1}{2}$-player game graph.}
Given the abstract state space $\Xh$ and the over and under-approximations of the transition functions $\Fu$ and $\Fo$, we are ready to construct the abstract  $2\sfrac{1}{2}$-player game graph induced by a CMP.

\begin{definition}
\label{def:reduced game}
	Let $\Sys$ be a given CMP. Then its induced abstract $2\sfrac{1}{2}$-player game graph is given by $\mathcal G = \tup{V,E,\tup{V_0,V_1,V_r}}$ such that
	\begin{itemize}
		\item $V_0 = \Xh$ and $V_1 = \Xh\times \Us$;
		\item $V_r=\bigcup_{v_1\in V_1}V_r(v_1)$, where\\
		$V_r(v_1) \!\coloneqq\! \set{v_r\!\subseteq\! \Xh \mid \Fu(v_1)\!\subseteq\! v_r \!\subseteq\! \Fo(v_1),1\!\le\!|v_r|\le |\Fu(v_1)| \!+\! 1};$
		\item and it holds that
		\begin{compactitem}[$\circ$]
		 \item for all $v_0\in V_0$, $E(v_0)=\set{(v_0,u)\mid u\in\Us}$
		 \item for all $v_1\in V_1$, $E(v_1)=V_r(v_1)$, and
		 \item for all $v_r\in V_r$, $E(v_r)=\{v_0\in V_0 \mid v_0\in v_r\}$.
		\end{compactitem}
	\end{itemize}
\end{definition}

Note that $V_r(v_1)$ contains non-empty subsets of $\Xh$ that includes all the abstract states in $\Fu(v_1)$ and possibly include only one additional element from $\Fo(v_1)$. The construction is illustrated in Fig.~\ref{fig:gadget}.

\begin{figure}
	\centering
 	\input{fig_gamereduction.tex}
	\caption{Illustration of the construction of the abstract $\twohalf$-player game (right) from a continuous-state CMP (left). The state space of the CMP is discretized into rectangular abstract states $A1,\hdots,D3$; $\Fu(A2,u) = \set{C2,C3}$ (intersecting the green region), and $\Fo(A2,u) = \set{C1,C2,C3,C4}$ (intersecting orange region). $V_0$, $V_1$ and $V_r$ are indicated by circle, rectangular and diamond-shaped nodes. Random vertices are dashed.
	}
	\label{fig:gadget}
\end{figure}

In the reduced game, $\p{0}$ models the controller, $\p{1}$ models the effect of discretization of the state space of $\mathfrak{S}$, and the random edges from the states in $V_r$ model the stochastic nature of the transitions of $\mathfrak{S}$.
Intuitively, the game graph in \REFdef{def:reduced game} captures the following interplay which is illustrated in \REFfig{fig:gadget}:
At every time step, the control policy for $\Sys$ has to choose a control input $u\in \Us$ based on the current vertex $\xh$ of $\game$.
Since the control policy is oblivious to the precise continuous state $\xs\in \Xs$ of $\Sys$, hence $u$ is required to be an optimal choice for \emph{every} continuous state $\xs\in \xh$.
This requirement is materialized by introducing a fictitious adversary (i.e.\ $\p{1}$) who, given $\xh$ and $u$, picks a continuous state $\xs\in \xh$ from which the control input $u$ is to be applied.
Now, we know that no matter what continuous $\xs$ is chosen by $\p{1}$, $\Ker(\Fu(\xh,u)\mid \xs,u)>\varepsilon$ holds for some $\varepsilon>0$.
This explains why every successor of the $(\xh,u)\in V_1$ states contains the set of vertices $\Fu(\xh,u)$.
Moreover, depending on which exact $\xs\in \xh$ $\p{1}$ chooses, with positive probability the system might go to some state in $\Fo(\xh,u)\setminus \Fu(\xh,u)$.
This is materialized by adding every state in $\Fo(\xh,u)\setminus \Fu(\xh,u)$ at a time to the successors of the states in $V_1$ (see \REFdef{def:reduced game}).
Finally, we assume that the successor from every state in $V_r$ is chosen uniformly at random (indicated by dotted edges in \REFdef{def:reduced game}).
Later, it will be evident that the exact probability values are never used for obtaining the almost sure winning region, and so we could have chosen any other probability distribution.

\smallskip
\noindent\textbf{Abstract parity specification.}
To conclude the abstraction of a given CMP $\Sys$ and its parity specification $\mathcal{P}=\set{B_1,\hdots,B_\ell}$, we have to formally translate the priority sets $B_i$ defined over subsets of states of the CMP into a partition of the vertices of the abstract $\twohalf$-player game graph $\game$ induced by $\Sys$. To this end, recall that we have assumed that the state space abstraction $\Xh$ respects the priority set $\mathcal{P}$.

\begin{definition}\label{def:abstract parity spec}
 Let $\Sys$ be a CMP with parity specification $\parity$ and $\game$ the abstract $\twohalf$-player game graph induced by $\Sys$. Then the induced abstract parity specification $\widehat{\mathcal{P}}=\set{\Bh_1,\ldots,\Bh_\ell}$ is defined such that
  $\Bh_i=\{v_0\in V_0~|~Q^{-1}(v_0)\subseteq B_i\}$ for all $i\in[1;\ell]$.
We denote the resulting $\twohalf$-player parity game by the tuple
 $\tup{\mathcal G,\widehat{\mathcal{P}}}$.
\end{definition}

 We note that the choice of the abstract parity set $\widehat{\mathcal{P}}$ does not partition the state space. Indeed, we implicitly assign an \enquote{undefined} color \enquote{$-$} to all nodes $V_1\cup V_r$. Thereby, we only interpret the given parity specification over a projection of a run to its player $0$ nodes. Formally, a run $r$ over the abstract game graph $\game$ starting from a vertex $\xs^0\in V_0$ is of the form:
\[
r\!=\!\xs^0,\!(\xs^0,u^0),\! (\set{\xs^{0,0}\!,\ldots,\!\xs^{0,i_0}}), \!\xs^1\!, (\xs^1\!,\!u^1),\! (\set{\xs^{1,0}\!,\ldots,\!\xs^{1,i_1}}), \ldots
\]
where $\xs^k\in \set{\xs^{k,0},\ldots,\xs^{k,i_k}}$ for all $k\in\nat$. The projection of the run $r$ to the $\p{0}$ states is defined as $\mathrm{Proj}_{V_0}(r)=\xs^0,\xs^1,\ldots$.
Let $\varphi$ be an $\omega$-regular specification defined using a set of predicates over $V_0$.
We use the convention that $(\game\models \varphi)$ will denote the set of every infinite run $r$ of $\game$, for any arbitrary pair of strategies of $\p{0}$ and $\p{1}$, such that $\mathrm{Proj}_{V_0}(r)$ satisfies $\varphi$.
This convention is well-defined because every infinite run of $\game$ will have infinitely many occurrences of vertices from $V_0$ in it:
This follows from the strict alternation of the vertices in $V_0$, $V_1$, and $V_r$, as per \REFdef{def:reduced game}.

\new{
\subsection{Preliminaries: Fixpoint Algorithms using $\mu$-Calculus}\label{sec:fixpt:prelim}

We explain the preliminary concepts of $\mu$-calculus, which is a powerful logical framework that can succinctly represent complex fixpoint algorithms involving set transformations \citep{Kozen83}.
Suppose $\mathcal{G}$ is a $\twohalf$-player game graph.
We will use two set transformers, namely the controllable predecessor operator $\cpre(\cdot)$ and the almost-sure predecessor operator $\apre(\cdot,\cdot)$, whose definitions are reiterated from our previous work \citep{banerjee2023fast}.
For any given sets of vertices $S,T\subseteq V$ of $\mathcal{G}$, define
\begin{align*}
	\cpre(S) &\coloneqq \set{v\in V_0 \mid E(v)\cap S\neq \emptyset} \cup \set{v\in V_1\cup V_r\mid E(v)\subseteq S},\\
	\apre(S,T)&\coloneqq \cpre(T) \cup \set{v\in V_r \mid E(v)\subseteq S \wedge E(v)\cap T\neq \emptyset }.
\end{align*}
Intuitively, the set $\cpre(S)$ is the set of vertices from which $\p{0}$ can force the game to reach $S$ in one step.
On the other hand, the set $\apre(S,T)$ is a subset of $\cpre(S)$ (i.e.\ the game remains in $S$ after one step) from which additionally there is a positive probability transition to the set $T$.

The standard syntax of $\mu$-calculus \citep{Kozen83} is given below.
The formulas of $\mu$-calculus, interpreted over the set of vertices of a $\twohalf$-player game graph $\game $,
are given by the grammar
\[
\varphi\; \Coloneqq \; p \mid X \mid \varphi \cup \varphi \mid \varphi \cap \varphi \mid \mathit{pre}(\varphi) \mid \mu X.\varphi \mid \nu X.\varphi
\]
where $p$ ranges over subsets of $V$, $X$ ranges over a set of formal variables,
$\mathit{pre}$ ranges over monotone set transformers in $\set{\cpre,\apre}$, and $\mu$ and $\nu$ denote, respectively, the least and the greatest fixed point of the functional
defined as $X \mapsto \varphi(X)$.
Since the operations $\cup$, $\cap$, and the set transformers $\mathit{pre}$ are all monotonic, the fixed points are guaranteed to exist.
A $\mu$-calculus formula evaluates to a set of states over $\game $, and the set can be computed
by induction over the structure of the formula, where the fixed points are evaluated by iteration.
We omit the standard semantics of formulas (see \cite{Kozen83}).
Instead, in Fig.~\ref{alg:illustration of mu-calculus}, we provide two small examples to illustrate how $\mu$-calculus formulas translate to symbolic fixpoint algorithms.
Each of the two examples in Fig.~\ref{alg:illustration of mu-calculus} has only a single $\mu$ or $\nu$ operator, but in general there can be multiple of them arranged in a nested fashion.
}

\begin{figure}
	\begin{minipage}{0.45\textwidth}
	\new{
		\begin{algorithmic}[1]
			\Require Game graph $\mathcal G = \tup{V,E,\tup{V_0,V_1,V_r}}$, target set $T\subseteq V$
			\Ensure $W=\mu X\;.\;T\cup \cpre(X)$
			\State $X^0\gets \emptyset$
			\ForRepeat{$i=0,1,2,\ldots$}
				\State $X^{i+1}\gets T\cup \cpre(X^i)$
			\Until{$X^{i+1}\neq X^{i}$}
			\State \Return $W\gets X^{\mathit{last}}$
		\end{algorithmic}
		}
	\end{minipage}
	\qquad
	\begin{minipage}{0.45\textwidth}
	\new{
		\begin{algorithmic}[1]
			\Require Game graph $\mathcal G = \tup{V,E,\tup{V_0,V_1,V_r}}$, safe set $S\subseteq V$
			\Ensure $W=\nu X\;.\;S\cap \cpre(X)$
			\State $X^0\gets V$
			\ForRepeat{$i=0,1,2,\ldots$}
				\State $X^{i+1}\gets S\cap \cpre(X^i)$
			\Until{$X^{i+1}\neq X^{i}$}
			\State \Return $W\gets X^{\mathit{last}}$
		\end{algorithmic}
		}
	\end{minipage}
	\caption{\new{Left: The fixpoint algorithm computing $\mu X\;.\; T\cup \cpre(X)$, which represents the set of vertices from which $\p{0}$ can enforce a visit to $T$, against every strategy of $\p{1}$.
	Right: The fixpoint algorithm computing $\nu X\;.\;S\cap \cpre(X)$, which represents the set of vertices from which $\p{0}$ can keep the game within $S $ forever, against every strategy of $\p{1}$.
	In each example, the variable $X^{\mathit{last}}$ represents the value of $X^i$ on the last iteration of the loop.
	For clarity reasons, we store all intermediate computation results in the $X^i$ sets, but in reality, we need to keep only the $X^i$ sets from the two latest iterations.
	Notice that, in both cases, the sequence $X^0,X^1,\ldots$ is monotone with respect to set inclusion, implying that the loop will always terminate.}}
	\label{alg:illustration of mu-calculus}
\end{figure}

\subsection{Abstract Controller Synthesis}\label{sec:abstractsynt}

Suppose $\tup{\mathcal G,\widehat{\mathcal{P}}}$ is the $\twohalf$-player parity game constructed from the CMP~$\Sys$ according to Def.~\ref{def:reduced game}.
Then solving the abstract controller synthesis problem amounts to computing the almost sure winning strategy $\pi_0$ of $\p{0}$ over $\tup{\mathcal G,\widehat{\mathcal{P}}}$.
In this section, we present a symbolic algorithm for this computation on $\twohalf$-player parity games.
\new{
While the established algorithm \citep{CJH03} solves $\twohalf$-player games through a reduction to $2$-player games, our algorithm is direct and faster---both in theory and in practice \citep{banerjee2023fast,banerjee2022tacas}.
}
Our symbolic algorithm for $\twohalf$-player parity games has already appeared as a special case of fair adversarial games in our recent paper \citep{banerjee2023fast}, though the proof of correctness is via a reduction to fair adversarial \emph{Rabin} games.
In this section, as one of our main contributions of this paper, we present a direct and simpler proof for the same algorithm.

Let us review the symbolic algorithm for $\twohalf$-player parity games from \citet{banerjee2023fast}, before we give our new proof.
Without loss of generality, we assume that the highest priority $\ell$ appearing in the abstract parity specification $\widehat{\mathcal{P}} = \set{\Bh_1,\ldots,\Bh_\ell}$ is \emph{even}.
(When this is not the case by default, we add a dummy set $\Bh_\ell=\emptyset$ with the next higher even priority $\ell$.)
With these, the fixpoint algorithm for solving parity games can be expressed using the following $\mu$-calculus expression (cf. \cite[Eq.(22)]{banerjee2023fast}):
\begin{align}\label{equ:2.5 parity}
& \nu Y_{\ell}.~\mu X_{\ell-1}\ldots ~\mu X_3.~\nu Y_{2}.~\mu X_1.\notag \\
                &\quad(\Bh_1 \cap \apre(Y_2,X_1)) \cup (\Bh_2 \cap \cpre(Y_2)) \cup \notag\\
                &\quad\quad \left((\Bh_1\cup\Bh_2\cup\Bh_3)\cap \apre(Y_4, X_3)\right) \cup (\Bh_4\cap \cpre(Y_4)) \cup \notag\\
                &\quad\quad\quad\ldots\notag\\
                &\quad\quad\quad\quad \left((\Bh_1\cup\ldots\cup\Bh_{\ell-1})\cap \apre(Y_\ell,X_{\ell-1})\right)\cup  (\Bh_{\ell} \cap \cpre(Y_{\ell} )).
\end{align}

In the rest of this section, we make the number of priorities explicit by using the notation $\widehat{\mathcal{P}}_\ell$ to denote the set of priorites $\set{\Bh_1,\ldots,\Bh_\ell}$ for some even $\ell$.

In order to obtain a direct proof of the above parity fixpoint for $\twohalf$-player games in Thm.~\ref{thm:correctness of 2.5 parity fp}, we use an argument similar to the one used by \citet{EJ91} in the correctness proof of their symbolic algorithm for solving \emph{$2$-player} parity games.
\new{
Define the set $R_\ell = \cup_{i=1}^\ell \Bh_\ell$.
We use an induction over the number of priorities $\ell$, and inductively decompose the parity specification $\widehat{\mathcal{P}}_\ell$ into two hierarchically connected sub-specifications.
Almost sure satisfaction of $\mathit{Parity}(\widehat{\mathcal{P}}_\ell)$ involves: (I) when in $R_\ell$, almost surely, either visit $\Bh_\ell$ infinitely often or eventually reach the almost sure winning region for $\mathit{Parity}(\widehat{\mathcal{P}}_{\ell-2})$, which is a subset of $R_{\ell-2}$, and (II) when in $R_{\ell-2}$, almost surely, either eventually reach the almost sure winning region for $\mathit{Parity}(\widehat{\mathcal{P}}_{\ell-2})$ or eventually reach the vertices in region $R_\ell$ from where (I) can be satisfied.
Clearly, there is a two-way dependence between (I) and (II), which is handled in the nested fixpoint computation by iteratively solving (I) and (II) one after the other until convergence.

The main difference between our fixpoint for $\twohalf$-player parity games and the fixpoint for $2$-player parity games is related to the component (I) above:
For $2$-player games, it is sufficient to restrict the movement inside the region $\Bh_\ell\cup \Bh_{\ell-1}$ while fulfilling (I), whereas for $\twohalf$-player games, $\p{0}$ needs to be given more freedom in visiting other lower priority vertices, as long as (I) can be fulfilled.
}

To easily reason about the nested dependence between (I) and (II) in our inductive proof of correctness, we ``flatten out'' the nested fixpoint computation by fixing the solution of (I) using a symbolic variable $T$, that itself gets updated subsequently based on the solution of (II).
A visualization of this decomposition of the fixpoint for $\mathit{Parity}(\widehat{\mathcal{P}}_{\ell})$ (as in \eqref{equ:2.5 parity}) into (I) and (II) using the variable $T$ is as follows:
\begin{center}
\begin{tikzpicture}
	\draw	[fill=red,fill opacity=0.1,thick,red]	(-0.5,0.75)	--	(10.8,0.75)	--	(10.8,4.1)	--	(2.2,4.1)	--	(2.2,3.5)	--	(-0.5,3.5)	--	cycle;
	\draw[fill=blue, fill opacity=0.1, thick,blue]	(0.5,-0.3)	--	(9.5,-0.3)	--	(9.5,0.7)	--	(0.5,0.7)	--	cycle;
	\node[align=center]	at	(5,2)
		{$\begin{aligned}
			& \nu Y_{\ell}.~\mu X_{\ell-1}.~\nu Y_{\ell-2}.~\mu X_{\ell-3}\ldots ~\mu X_3.~\nu Y_{2}.~\mu X_1. \\
                &\quad(\Bh_1 \cap \apre(Y_2,X_1)) \cup (\Bh_2 \cap \cpre(Y_2)) \cup \notag\\
                &\quad\quad \left((\Bh_1\cup\Bh_2\cup\Bh_3)\cap \apre(Y_4, X_3)\right) \cup (\Bh_4\cap \cpre(Y_4)) \cup \notag\\
                &\quad\quad\quad\ldots\notag\\
                &\quad\quad\quad\quad  \left((\Bh_1\cup\ldots\cup\Bh_{\ell-3})\cap \apre(Y_{\ell-2},X_{\ell-3})\right)\cup  (\Bh_{\ell-2} \cap \cpre(Y_{\ell-2} ))\,\cup \\
                &\quad\quad\quad\quad \left((\Bh_1\cup\ldots\cup\Bh_{\ell-1})\cap \apre(Y_\ell,X_{\ell-1})\right)\cup  (\Bh_{\ell} \cap \cpre(Y_{\ell} )).
		\end{aligned}$};
	\node[align=center,draw]		(a)	at	(9.5,3.5)	{Fixpoint for \\$\mathit{Parity}(\widehat{\mathcal{P}}_{\ell-2})$};
	\node	(b)	at	(11,0.2)	{$T$};
	\draw[->]	(b.west)		--	(9.5,0.2);
\end{tikzpicture}
\end{center}

For fixed values of $Y_{\ell}$ and $X_{\ell-1}$, the value of $T$ (i.e.\ part (I)) is fixed, which is updated subsequently based on the outcome of the fixpoint for $\mathit{Parity}(\widehat{\mathcal{P}}_{\ell-2})$ (i.e.\ part (II)).
Using induction over $\ell$, we show that the above fixpoint computes exactly the set of vertices from where (I) or (II) can be fulfilled almost surely.

The base case of the inductive proof is stated in the following lemma.
Here we consider the case $\ell = 2$ along with the set $T$ that we informally talked about earlier.
This case can be seen as a B\"uchi-or-reachability specification.

\begin{lemma}\label{lem:a.s. safe buchi or reach}
	Let $B,T\subseteq V$ be two sets of vertices of a $\twohalf$-player game graph $\game$, and let $B^c$ denote the complement of $B$.
	Then the fixed point of the following $\mu$-calculus expression is $\p{0}$'s almost sure winning region for the winning condition $\square\lozenge B\lor \lozenge T$:
	\begin{equation}\label{equ:a.s. reach-avoid}
		\nu Y.~\mu X.~(B^c\cap \apre(Y,X)) \cup (B\cap \cpre(Y)) \cup T.
	\end{equation}
\end{lemma}

For the proof, we will use some concepts which need to be defined first.
Let $\pi_0\in\Pi_0^{\mathrm{DM}}$ and $\pi_1\in\Pi_1^{\mathrm{DM}}$ be a pair of arbitrary deterministic memoryless strategies of the two players.
Fixing $\pi_0$ and $\pi_1$ in the game graph $\game$ induces a Markov chain $M=\tup{V',E',\tup{\emptyset,\emptyset,V_r'}}$ with $V'=V_r'=V$ and $E'\subseteq E$ where every $v\in V_0\cup V_1$ has a single successor in $E'$ given by $\pi_0(v)$ or $\pi_1(v)$ as applicable.
Concretely, for every $v\in V_0$, $(v,v')\in E'$ iff $\pi_0(v)=v'$ and for every $v\in V_1$, $(v,v')\in E'$ iff $\pi_1(v)=v'$.
We call $M$ the \emph{Markov chain induced on $\game$ by $\pi_0$, $\pi_1$}.
We will also use the concept of Bottom Strongly Connected Components (BSCC) of a given Markov chain $M = \tup{V,E,\tup{\emptyset,\emptyset,V_r}}$:
A set of vertices $C\subseteq V$ of $M$ is called a BSCC if (1) the subgraph $\tup{C,E\cap (C\times C)}$ is a strongly connected component of the graph $\tup{V,E}$, and (2) no vertex in $C$ has an outgoing edge to a vertex outside $C$, i.e.\ $E(C)\subseteq C$.

\begin{proof}[Proof of Lem.~\ref{lem:a.s. safe buchi or reach}]
We prove Lem.~\ref{lem:a.s. safe buchi or reach} in two parts:
	\begin{description}
		\item[Soundness:]
	Suppose $Y^*$ is the fixed point of the expression \eqref{equ:a.s. reach-avoid}.
	Consider the computation of the inner fixpoint over $X$ when the outer fixpoint has reached the last iteration $Y^*$.
	We get a sequence of $X^i$-s as in the following:
	\begin{align}\label{equ:a.s. reach-avoid mu-calculus iterations}
		X^0 &= \emptyset\notag\\
		X^1 &= (B\cap \cpre(Y^*)) \cup T\notag\\
		X^2 &= (B^c\cap \apre(Y^*,X^1)) \cup (B\cap \cpre(Y^*)) \cup T\notag\\
		X^3 &= (B^c\cap \apre(Y^*,X^2)) \cup (B\cap \cpre(Y^*)) \cup T\notag\\
		&\vdots\notag\\
		Y^* = X^p &= (B^c\cap \apre(Y^*,X^p)) \cup (B\cap \cpre(Y^*)) \cup T,
	\end{align}
	where $X^p$ is the fixed point of the inner fixpoint equaton.
	Observe that $X^0\subseteq X^1\subseteq X^2\subseteq \ldots \subseteq X^p$.
	Since we do not care what happens to the runs after they reach $T$, let us assume that every vertex in $T$ is a sink vertex with a self loop on it; clearly every vertex in $T$ forms a BSCC on its own.
	Let $v$ be any arbitrary vertex in $Y^*$.
		From \eqref{equ:a.s. reach-avoid mu-calculus iterations} it follows that there exists a strategy $\pi_0$ of $\p{0}$, such that for every strategy $\pi_1$ of $\p{1}$, there is a run of positive probability of length $\leq p$ that reaches either $B$ or $T$.

		Moreover, $\pi_0$ ensures that no run from $v$ goes outside $Y^*$ even after reaching $B$, so that with positive probability either $B$ can be reached again or $T$ can be eventually reached.
		These imply that every path in the Markov chain, induced on $\game$ by $\pi_0$ and an arbitrary $\p{1}$ strategy $\pi_1$, will eventually reach a BSCC that either intersects with $B$ or is in $T$, almost surely.
		From the ergodicity property of Markov chains (see \citet[Prop.~4.27, 4.28]{kemenydenumerable}) it follows that $P_v^{\pi_0,\pi_1}\left(\game\models \square\lozenge B \lor \lozenge T\right)=1$.
		\item[Completeness:] Let $Y^0=V\supseteq Y^1\supseteq Y^2\supseteq \ldots \supseteq Y^q$ be the sequence of sets of vertices represented by the $Y$ fixpoint variable, where $q$ is the iteration index where the fixpoint converges (i.e.\ the minimum index $q$ such that $Y^q=Y^{q+1}$).
		We show, by induction, that for every iteration index $i\in [0;q]$, the almost sure winning region for the parity specification is contained in $Y^i$.

		The base case is when $i=0$, in which case the claim follows trivially by observing that $Y^0=V$.
		Suppose the claim holds for iteration $i$, i.e.\ $Y^i$ contains the almost sure winning region and every vertex outside $Y^i$ is losing with positive probability.
		When the iterations over $Y$ reach $i+1$, for every vertex $v$ in $Y^{i+1}\setminus Y^{i}$, for every $\p{0}$ strategy there is a $\p{1}$ strategy such that
		either (a) the plays exit $Y^i$ in one step, surely or with positive probability, or (b) no positive probability path exists from $v$ that reaches $B\cap \cpre(Y^i)$ or $T$.
		(Both follow from the definition of $\cpre$ and $\apre$ and the fixpoint expression.)

		When (a) is true, $v$ is not almost surely winning because every vertex outside $Y^i$ is losing with positive probability by induction hypothesis.
		When (b) is true, $v$ is not almost surely winning because it is no longer possible to reach a vertex in $B$ from which the play surely does not leave $Y^i$.
		Thus $v$ cannot belong to the almost sure winning region, implying our induction claim.

	\end{description}
\end{proof}

Now we establish the correctness of our proposed parity fixpoint given in \eqref{equ:2.5 parity}.

\begin{theorem}\label{thm:correctness of 2.5 parity fp}
	The fixpoint \eqref{equ:2.5 parity} is a sound and complete algorithm for computing the almost sure winning region for the parity winning condition $\widehat{\mathcal{P}}=\set{\Bh_1,\ldots,\Bh_\ell}$.
	\new{If $n$ is the number of vertices of the abstract $\twohalf$-player game graph $\mathcal{G}$, then \eqref{equ:2.5 parity} runs in $\ell\cdot n^{\log(\ell/\log(n)) + \mathcal{O}(1)}$ symbolic steps, and uses $\mathcal{O}(n\cdot  \polylog(n))$ space.\footnote{The notation $\polylog(n)$ represents some polynomial of $\log(n)$.} }
\end{theorem}

The proof of Thm.~\ref{thm:correctness of 2.5 parity fp} uses an induction over the number of priorities $\ell$.
The soundness claim uses the same ergodicity property of Markov chains, as in the proof of Lem.~\ref{lem:a.s. safe buchi or reach}, to establish that either the $\Bh_\ell$ set is visited infinitely often or eventually the almost sure winning region for $\widehat{\mathcal{P}}_{\ell-2}$ is reached almost surely.
The completeness claim uses an induction over the indices of the outermost fixpoint variable $Y_\ell$ as in the proof of Lem.~\ref{lem:a.s. safe buchi or reach}, and shows that every almost sure winning vertex will remain included in every iteration over $Y_\ell$.
\new{
The time and space complexity follows from the recent advances in fixpoint algorithms \citep{arnold2021quasi}.
}
The complete proof is included in the appendix.

\new{
\begin{remark}\label{remark:complexity comparison with coogan}
	We compare the computational complexity of our controller synthesis approach against the state-of-the-art approach by \cite{Coogan20}, which uses a BMDP abstraction and an enumerative synthesis procedure to solve the problem.
	For the abstraction stage, for every abstract state-input pair, our main computation is in the computation of reachable sets, whereas their main computation is in computing upper and lower bounds of the transition probabilities (through optimizations).
	Although our approach can be computationally more efficient for the broad class of mixed-monotone systems, their approach generates richer abstractions that embed transition probability bounds in them, which is useful for solving the quantitative synthesis problem that we do not support.
	Furthermore, their lazy abstraction-refinement heuristic reduces the computational burden in cases where finer partition elements are required only in a smaller region of the state space.

	For the synthesis stage, our symbolic algorithm has a clear complexity-theoretic advantage, compared to their enumerative algorithm.
	For fair comparison, we first analyze the time complexity of our algorithm when implemented in an enumerative fashion.
	It can be observed that the number of states $n$ in the abstract $\twohalf$-player game is given by $\mathcal{O}(|\Xh|\cdot |\Us| + |\Xh|^2)$, which simplifies to $\mathcal{O}(|\Xh|^2)$ when $|\Xh|\geq |\Us|$ (the usual case).
	Then the time-complexity of \eqref{equ:2.5 parity} is given as $\ell\cdot |\Xh|^{2\log(\ell/\log( |\Xh|^2)) + \mathcal{O}(1)}$ symbolic steps.
	The enumerative implementation of the predecessor operators requires iterating over every edge of the game graph at most once, and for every edge, iterating over every vertex at most once (to check if the end point of the edge belongs to a given set).
	This takes $\mathcal{O}(n^3) = \mathcal{O}(|\Xh|^6)$-time in the worst case.
	Therefore, in terms of absolute time, the time-complexity of \eqref{equ:2.5 parity} is still given as $\ell\cdot |\Xh|^{2\log(\ell/\log( |\Xh|^2)) + \mathcal{O}(1)}$.

	On the other hand, the technique by \cite{Coogan20} takes $\mathcal{O}(2^{|\Xh|})$ time, which is significantly higher compared to our approach.
	In practice, the symbolic implementation of our predecessor operators is usually way more efficient and succinct than the enumerative implementation.

	\cite{Coogan20} have not discussed the space complexity of their algorithm, and it is also non-trivial to be deduced.
	Empirically, our approach took significantly less memory compared to their approach; see Sec.~\ref{sec:case_study}.
\end{remark}

The time-complexity analysis of our algorithm in Rem.~\ref{remark:complexity comparison with coogan} can be slightly improved.
The quadratic blow-up due to the blow-up in the size of the abstract $\twohalf$-player game originates from the explicit representation of the gadget in Fig.~\ref{fig:gadget}.
In reality, this blow-up can be avoided through an implicit construction of the $\twohalf$-player game graph, where the predecessor operators implicitly construct the gadget, locally, from the functions $\Fo$ and $\Fu$; see the work by \cite{majumdar2020symbolic} for details.
In terms of symbolic steps, the resulting time-complexity gets reduced to $\ell\cdot |\Xh|^{\log(\ell/\log( |\Xh|^2)) + \mathcal{O}(1)}$.
}

\subsection{Controller Refinement }\label{sec:refine}
Now consider that the abstract $2\sfrac{1}{2}$-player parity game $\tup{\mathcal G,\widehat{\mathcal{P}}}$ constructed from the CMP~$\Sys$ via \REFdef{def:reduced game} and \REFdef{def:abstract parity spec} has been solved as discussed in \REFsec{sec:abstractsynt}. Hence, we know the almost-sure $\p{0}$ winning region $\mathcal{W}(\game\models \paritya)$ and the associated deterministic memoryless almost-sure $\p{0}$ winning policy $\pi_0\in \Pi^\DM$. Then we refine $\pi_0$ to a control policy $\rho\in\Policy$ for the CMP~$\Sys$ under parity condition $\paritya$ as follows.

\begin{definition}
 Let $\Sys$ be a CMP with parity specification $\parity$ and $\tup{\mathcal G,\widehat{\mathcal{P}}}$ its induced finite $2\sfrac{1}{2}$-player parity game with deterministic memoryless almost sure $\p{0}$ winning strategy $\pi_0\in \Pi^\DM$. Then the control policy $\rho\in\Policy$ is called \emph{the refinement of} $\pi_0$ if and only if
for every $\xs\in \Xs$, if $\xs\in \xh$ for some $\xh\in \Xh$, and if $\pi_0(\xh)=(\xh,u)\in V_1$ for some $u\in \mathcal{U}$, then $\Cont(\xs)\coloneqq u$.
\end{definition}

With the completion of this last step of our ABCD method for stochastic nonlinear systems we can finally state our main theorem providing a solution to Problem~\ref{prob:policy}, which
we prove in \REFsec{sec:soundness}.

\begin{theorem}[Solution of Problem~\ref{prob:policy}]\label{thm:main}
	Let $\Sys$ be a CMP and $\mathit{Parity}(\mathcal{P})$ be a given parity specification.
	Let $\tup{\game,\widehat{\mathcal{P}}}$ be the abstract $2\sfrac{1}{2}$-player game defined in \REFdef{def:reduced game}.
	Suppose, a vertex $\xh\in V_0$ is almost sure winning for $\p{0}$ in the game $\tup{\game,\widehat{\mathcal{P}}}$, and $\pi_0\in \Pi^\DM$ is the corresponding $\p{0}$ winning strategy.
	Then the refinement $\Cont$ of $\pi_0$ ensures that $\xh\subseteq \W(\Sys,\Cont)$.
\end{theorem}

\begin{remark}\label{remark:quality of approximation}
An \emph{over}-approximation of the optimal almost sure winning domain $\W^\ast(\Sys)$ of $\Sys$ w.r.t.\ $\parity$ can be computed via
$\tup{\game,\widehat{\mathcal{P}}}$ as well.
To obtain an \emph{over}-approximation, we solve this abstract game \emph{cooperatively}.
That is, we let player $\p{0}$ choose both its own moves and the moves of player $p_1$ to win almost surely with respect to $\paritya$.
Then the approximation error of Problem~\ref{prob:policy} can be upper-bounded as the Lebesgue measure of the set difference of the over- and the under-approximation, as will be done for the experiments in \REFsec{sec:case_study}.
\end{remark}

\new{
\begin{remark}\label{rem:on continuous inputs}
	Recall that we assumed that the control input space $\Us$ of $\Sys$ is finite.
	If $\Us$ is picked from a continuous space, we can discretize the set $\Us$ to obtain a  finite abstraction, similar to how we obtained a finite abstraction of the state space $\Xs$.
	The theoretical results of this paper are still valid, i.e., our algorithm produces a control policy which solves Prob.~\ref{prob:policy}.
	As expected, the winning domain of the obtained control policy will get larger as the input space discretization gets finer.
	But for finer input space discretization, the size of the abstraction increases resulting in higher computational cost both in the abstraction (more time needed to compute $\Fo,\Fu$) and in the synthesis stage (the size of the $\twohalf$-player game gets larger).
\end{remark}
}

%% file: fig_gamereduction.tex
\tikzstyle{decision} = [diamond, draw=gray, text badly centered, inner sep=1pt, aspect=2.5,line width=0.7mm]

\begin{tikzpicture}[scale=0.5]
		\newcommand{\offset}{4}

		\draw[fill=blue!60]	($(0,2)+(\offset,\offset)$)	rectangle	($(1,3)+(\offset,\offset)$);
		\draw[fill=blue!60]	($(0,2)+(\offset,\offset)$)	rectangle	($(1,3)+(\offset,\offset)$);
		\draw[draw=orange!80,fill=orange!80]	($(2.1,0.4)+(\offset,\offset)$)	rectangle	($(2.8,3.3)+(\offset,\offset)$);
		\draw[draw=green!60!black,fill=green!60!black]	($(2.3,1.5)+(\offset,\offset)$)	rectangle	($(2.6,2.8)+(\offset,\offset)$);
		\draw[step=1,line width=0.3mm,black]	($(0,0)+(\offset,\offset)$)	grid	($(4,4)+(\offset,\offset)$);

		\node	at		($(-0.3,0.5)+(\offset,\offset)$)	{{\color{red} $4$}};
		\node	at		($(-0.3,1.5)+(\offset,\offset)$)	{{\color{red} $3$}};
		\node	at		($(-0.3,2.5)+(\offset,\offset)$)	{{\color{red} $2$}};
		\node	at		($(-0.3,3.5)+(\offset,\offset)$)	{{\color{red} $1$}};
		\node	at		($(0.5,4.3)+(\offset,\offset)$)	{{\color{red} $A$}};
		\node	at		($(1.5,4.3)+(\offset,\offset)$)	{{\color{red} $B$}};
		\node	at		($(2.5,4.3)+(\offset,\offset)$)	{{\color{red} $C$}};
		\node	at		($(3.5,4.3)+(\offset,\offset)$)	{{\color{red} $D$}};

        \begin{footnotesize}
         \node[circle,draw=gray,line width=0.7mm,fill=blue!20]	(A)	at		(16,9.5)	{$A2$};
         \node[rectangle,draw=gray,line width=0.7mm]		(B)	at		(16,7.5)	{ $(\textcolor{blue!60}{A2},u)$};
          \node[decision,draw]		(C)	at		(12,6)	{\scriptsize $\set{\textcolor{green!60!black}{C2,C3}}$};
          \node[decision,draw]		(D)	at		(16,5.3)	{\scriptsize $\set{\textcolor{orange}{C1},\textcolor{green!60!black}{C2,C3}}$};
          \node[decision,draw]		(E)	at		(20,6)	{\scriptsize $\set{\textcolor{green!60!black}{C2,C3},\textcolor{orange}{C_4}}$};
           \node[circle,draw=gray,line width=0.7mm,fill=orange!20]		(F)	at		(11,3.6)	{$C1$};
           \node[circle,draw=gray,line width=0.7mm,fill=green!60!black!20]		(G)	at		(13,3.6)	{$C2$};
           \node[circle,draw=gray,line width=0.7mm,fill=green!60!black!20]		(H)	at		(19,3.6)	{$C3$};
           \node[circle,draw=gray,line width=0.7mm,fill=orange!20]		(I)		at		(21,3.6)	{$C4$};

           \path[->,line width=0.3mm]	(A)	edge		(B)
                        (B)	edge		(C)
                                edge		(D)
                                edge		(E);

           \path[->,line width=0.3mm, dotted]	(C)	edge	(G)
                                edge[bend right=20] 		(H);

           \path[->,line width=0.3mm, dotted]  (D)	edge[bend right=20]		(F)
                                edge		(G)
                                edge[bend left=20]		(H);

          \path[->,line width=0.3mm,dotted]  (E)	edge[bend left=20]		(G)
                                edge		(H)
                                edge		(I);

        \end{footnotesize}
	\end{tikzpicture}

%% file: soundness.tex

\section{Proof of Theorem~\ref{thm:main}}
\label{sec:soundness}

\smallskip
\noindent{\textbf{Proof outline.}}
To prove \REFthm{thm:main}, we first decompose both the original and the abstract parity specifications $\parity$ and $\paritya$ into a combination of more manageable safety and reachability sub-parts. That is, for every state reachable by a finite run in  $\Sys$, we consider a local safety specification $\psiS$  and a local reachability specification $\psiR$ defined by
\begin{equation}
 \psiS: = \square \lnot B_i \text{ and } \textstyle\psiR: = \lozenge \left( \psiS \vee \bigvee_{j\ineven\incone{i+1}{k}} B_j \right).
\end{equation}
Intuitively, $\psiR$ requires that every time an odd priority---say $B_i$---is visited in $\Sys$, eventually either $B_i$ should never occur or an even priority $B_j$ with $j>i$ should occur, almost surely.
Similarly, for the abstract  $2\sfrac{1}{2}$-player parity game $\tup{\mathcal G,\widehat{\mathcal{P}}}$ we consider the local safety specification $\psihS$ and a local reachability specification $\psihR$ defined by
\begin{equation}
 \psihS := \square \lnot \Bh_i \text{ and } \psihR := \textstyle\lozenge \left( \psihS \vee \bigvee_{j\ineven\incone{i+1}{k}} \Bh_j \right).
\end{equation}

While the above decomposition needs to be established both for $\game$ and for $\Sys$, the directions of the respective proof differ.
For $\Sys$ we show that if $\psiR$ holds for a state reachable by a finite run over $\Sys$, then the original specification $\parity$ is satisfied by a continuation of the run using the refined policy $\rho$ (\textbf{Step 1}).
For $\game$ we show that if $\paritya$ is satisfied, then $\psihR$ holds for every state visited by a run compatible with the almost-sure winning strategy $\pi_0$ in $\tup{\mathcal G,\widehat{\mathcal{P}}}$ (\textbf{Step 2}).
Further, we show that satisfaction of $\psihS$ (resp. $\psihR$) in $\game$ implies satisfaction of $\psiS$ (resp. $\psiR$) in $\Sys$ (\textbf{Step 3-5}). With this, we have all ingredients to prove \REFthm{thm:main} (\textbf{Step 6}).

\smallskip
\noindent{\textbf{Step 1: Decomposition of $\parity$.} 
We prove a \emph{sufficient} condition for satisfaction of $\parity$ in $\Sys$ if $\psiR$ holds.

\begin{lemma}\label{lem:sufficient condition for parity in CMP}
	Let $\Sys$ be a CMP, $\parity$ be a parity specification, $\xs^0\in \Xs$ be a given initial state, and $\rho$ be a control policy.
	Suppose the following holds for every finite path $(\xs^0,\ldots, \xs^n)\in \Xs^{n+1}$ of $\Sys$ and every $i\inodd \incone{1}{k}$:
	\begin{align}\label{eq:reachability decomposition a}
		\xs^n\in B_i \implies P_{\xs^n}^{\rho} \left( \Sys\models
		\psiR
		\right) = 1.
	\end{align}
	Then 
		$P_{\xs^0}^{\rho}(\Sys\models \parity) = 1$.
\end{lemma}

\begin{proof}[Proof of Lem.~\ref{lem:sufficient condition for parity in CMP}]
 Define for any arbitrary $i\inodd \incone{1}{k}$ the event $E_i \coloneqq (\Sys\models\psi_i)$ with the specification $\psi_i \coloneqq \left(\square\lozenge B_{i} \!\implies\! \bigvee_{j\ineven \incone{i+1}{k}} \square\lozenge B_{j}\right)$.
	We want to show that
		$P_{\xs^0}^{\rho}(\Sys\models \parity)
		\textstyle= P_{\xs^0}^{\rho}\left(\bigcap_{i} E_i\right) = 1$
	where $i\inodd \incone{1}{k}$.
	We prove this by showing $P_{\xs^0}^{\rho}(\overline{E_i}) = 0$ for every $i\inodd \incone{1}{k}$. 
	Once we show this, the result follows according to the standard inequalities:
	\begin{align*}
	& P_{\xs^0}^{\rho}\textstyle \left(\bigcap_{i} E_i\right) =
	1- P_{\xs^0}^{\rho}\left(\bigcup_{i} \overline{E_i}\right)
	 \ge 1- \sum_{i}P_{\xs^0}^{\rho}\left(\overline{E_i}\right) = 1\\
	&\hspace{-0.5cm}\text{where}~P_{\xs^0}^{\rho}(\overline{E_i}) = P_{\xs^0}^{\rho}\left((\Sys\models \square\lozenge B_i)\cap (\Sys\models \wedge_{j} \lozenge\square \lnot B_j)\right),
	\end{align*}
	with $i\inodd \incone{1}{k}$ and $j\ineven\inctwo{i+1}{k}$.
	Define the random variable $\tau$ to be the largest time instance when the trajectory visits one of the sets $B_j$; \new{if none of the $B_j$ sets is ever visited, then $\tau$ is zero.} 
	Also define $\tau'>\tau$ to be the first time instance after $\tau$ when the trajectory visits $B_i$ again. 
	Note that for any trajectory satisfying $\square\lozenge B_i$ and $ \wedge_{j} \lozenge\square \lnot B_j$, both $\tau$ and $\tau'$ are well-defined and bounded. According to the assumption \eqref{eq:reachability decomposition a}, we have
	\begin{align*}
	P_{\xs^0}^{\rho}(\overline{E_i}\,|\, & \tau' = n, s^0,s^1,\cdots,s^n) \\
	& = P_{\xs^n}^{\rho}\left((\Sys\models \square\lozenge B_i)\cap (\Sys\models \wedge_{j}\lozenge\square \lnot B_j)\right) = 0.
	\end{align*}
	By taking the expectation with respect to the condition $(\tau',s^0,\cdots,s^n)$, we conclude that 
 	\begin{align*}
	 P_{\xs^0}^{\rho}(\overline{E_i}) \!=\! \mathbb E_{\xs^0}^{\rho} \left[ P_{\xs^0}^{\rho}(\overline{E_i}\,|\, \tau' = n, s^0,s^1,\cdots,s^n)\right] \!=\! \mathbb E_{\xs^0}^{\rho}[0] \!=\! 0,
 	\end{align*}
i.e., $\overline{E_i}$ has a zero probability.
\end{proof}

\smallskip
\noindent{\textbf{Step 2: Decomposition of $\paritya$.} 
We present a \emph{necessary} condition for satisfaction of $\paritya$ in $\game$ if $\psihR$ holds.

\begin{lemma}\label{lem:decompose proof of liveness}
	Let $\tup{\game,\widehat{\mathcal{P}}}$ be a $2\sfrac{1}{2}$-player parity game, and $v^0$ be a given vertex of $\game$.
	Suppose $\pi_0^*\in \Pi_0^\DM$ is a $\p{0}$ strategy such that $\inf_{\pi_1\in \Pi_1} P_{v^0}^{\pi_0^*,\pi_1}(\game\models \paritya) =1$.
	Then given every finite path $v^0\ldots v^n\in V^*$ such that there exists a $\p{1}$ strategy $\pi_1\in \Pi_1$ with $ P_{v^0}^{\pi_0^*,\pi_1}(\game\models v^0\ldots v^n) > 0$, the following holds	for every $i\inodd\inctwo{1}{k}$:
	\begin{align}
	\label{eq:reachability decomposition}
		v^n\in \Bh_i \implies \inf_{\pi_1\in \Pi_1}P_{v^n}^{\pi_0^*,\pi_1} \left( \game\models\lozenge 
\psihR
		\right) = 1.
	\end{align}
\end{lemma}

The only new factor in Eq.~\eqref{eq:reachability decomposition} is the presence of the adversarial effect of the $\p{1}$ strategies.

\begin{proof}
	It follows from the definition of the parity specification in \eqref{eq:definition of parity} that a vertex $v^0\in V$ is almost sure winning using the strategy $\pi_0^*$ if the following condition is fulfilled:
	\begin{align}\label{eq:parity alternate 1}\textstyle
		\inf_{\pi_1\in \Pi_1}P_{v^0}^{\pi_0^*,\pi_1} \left(\game\models \bigwedge_{i\inodd \inctwo{1}{k}} \square \left( \Bh_i \implies \lozenge \psiR\right)\right) = 1.
	\end{align}
	From the semantics of LTL, \eqref{eq:parity alternate 1} implies:
	\begin{align}\textstyle
		\inf_{\pi_1\in \Pi_1}P_{v^0}^{\pi_0^*,\pi_1} \left(\game\models  \bigwedge_{i} \forall m\in \mathbb{N}\;.\;(v^m\in \Bh_i) \implies \lozenge \psihR\right)	= 1,
		\label{eq:parity alternate}
	\end{align}
	with $i\inodd \inctwo{1}{k}$. We show that \eqref{eq:parity alternate} implies for every finite run $v^0\ldots v^n\in V^*$, occurring with a positive probability $p_1>0$ for some strategy of $\p{1}$, \eqref{eq:reachability decomposition} holds.
	Suppose, for contradiction's sake, there exists some $i\inodd \inctwo{1}{k}$ such that $v^n\in \Bh_i$ and \eqref{eq:reachability decomposition} does not hold, i.e.,
		 $\inf_{\pi_1\in \Pi_1}P_{v^n}^{\pi_0^*,\pi_1} \left( \game\models\lozenge \psihR\right) < 1$,
	implying existence of some $0 < p_2 \leq 1$ with
		 $\sup_{\pi_1\in \Pi_1}P_{v^n}^{\pi_0^*,\pi_1} \left( \game\not\models\lozenge \psihR\right) = p_2$.
	This results in satisfaction of the parity specification with a probability of \emph{at most} $(1- p_1\cdot p_2) < 1$, contradicting \eqref{eq:parity alternate}.
\end{proof}

\smallskip
\noindent{\textbf{Step 3: Refinement of $\psihS$ to $\psiS$.} 
We show that almost sure safety with respect to a given set $U$ in $\game$ implies the same with respect to the set $Q^{-1}(U)$ in $\Sys$; this will later be used to infer $\psihS\implies \psiS$.

\begin{proposition}\label{prop:sound safety}
	Let $\Sys$ be a CMP and $\game$ be a finite $2\sfrac{1}{2}$-player game graph as defined in \REFdef{def:reduced game}.
	Suppose $U\subseteq V_0$ is a given set of vertices of $\game$, and assume that there is a $\p{0}$ vertex $v\in U$ for which there is a strategy $\pi_0\in \Pi_0^{\mathrm{DM}}$ of $\p{0}$ such that $\inf_{\pi_1\in \Pi_1} P_v^{\pi_0,\pi_1}(\game\models \square U) = 1$.
	Then the refinement $\Cont\in\Policy$ of $\pi_0$ ensures that for every state $s\in v$, $P_s^\Cont(\Sys\models \square Q^{-1}(U))=1$.
\end{proposition}

\begin{proof}
	It is known that for safety properties, almost sure satisfaction coincides with sure satisfaction, i.e., $\inf_{\pi_1\in \Pi_1} P_v^{\pi_0,\pi_1}(\game\models \square U) = 1$ if and only if for every strategy $\pi_1\in \Pi_1$, every infinite run of $\game$ stays inside $U$ at all time \citep{dAH00}.
	In other words, there must be a controlled invariant set $W$ inside $U$ for the strategy $\pi_0$, and $v\in W$.
	This controlled invariant set can be obtained by considering the $2$-player game, obtained from $\game$ by removing all the random vertices, and redirecting the outgoing transitions of a given $\p{1}$ vertex $v'\in V_1$ to the $\p{0}$ vertices within the set $\Fo(v',\pi_0(v'))\subseteq V_0$.
	Since $\Fo(v',\pi_0(v'))$ overapproximates the set of all the continuous states reachable from $v'$ using the input $\pi_0(v')$, hence if $\p{0}$ can fulfill $\square U$ using the strategy $\pi_0(v')$, then $\rho$ can fulfill $\square Q^{-1}(U)$ from every state $\xs\in v'$ in $\Sys$.
	(This follows from the standard arguments in abstraction-based control using over-approximation based abstractions \citep{ReissigWeberRungger_2017_FRR}.)
\end{proof}

\smallskip
\noindent{\textbf{Step 4: Refinement of $\psihR$ to $\psiR$.} 
We show that almost sure reachability with respect to a given set $U$ in $\game$ implies the same with respect to the set $Q^{-1}(U)$ in $\Sys$; this will be used to infer $\psihR\implies \psiR$.

Let $\Sys$ be a CMP.
Let us consider a reachability specification $\lozenge U$, for a set $U\subseteq V_0$, in the game $\game$ defined in \REFdef{def:reduced game}.
Suppose $\pi_0\in \Pi_0^\DM$ is some strategy of $\p{0}$.

We introduce a ranking function $r\colon V_0\to \mathbb{N}\cup \set{\infty}$ as a certificate for almost sure satisfaction of the specification $\lozenge U$.
The ranking function $r$ is defined inductively as follows:
\begin{align}\label{eq:ranking function}
	r(v) \!=\! \begin{cases}
				0	&	\!\!\!v\in U,\\
				\infty	&  \!\!\!\inf_{\pi_1\in \Pi_1} P_v^{\pi_0,\pi_1}(\game\models \lozenge U) < 1,\\
				i+1	&	\!\!\!\min \set{n\in \mathbb{N} | \inf_{\pi_1\in \Pi_1} \!P_v^{\pi_0,\pi_1}\!(\game\models \bigcirc r^{-1}(n))\!>\!0}\\
				&\quad  =i \wedge \inf_{\pi_1\in \Pi_1} P_v^{\pi_0,\pi_1}(\game\not\models \bigcirc r^{-1}(\infty))=1.	
			\end{cases}
\end{align}

Note that every vertex $v\in V$ gets a rank:
If $r(v)\neq \infty$, then $\inf_{\pi_1\in \Pi_1} P_v^{\pi_0,\pi_1}(\game\models \lozenge U) =1$ by definition of $r$.
In this case, there must exist some path to $U$, i.e., $\inf_{\pi_1\in \Pi_1} P_v^{\pi_0,\pi_1}(\game\not\models \bigcirc r^{-1}(\infty))=1$ must be true, and moreover $ \inf_{\pi_1\in \Pi_1} P_v^{\pi_0,\pi_1}(\game\models \bigcirc r^{-1}(n))>0$ will be true for some $n$.
Thus, $r(v)=n+1$.

From the ranking function $r(\cdot)$ defined in \eqref{eq:ranking function}, it is clear that $\inf_{\pi_1\in \Pi_1} P_v^{\pi_0,\pi_1}(\game\models \lozenge U) = 1$ implies $r(v)\neq \infty$.
We first identify some local structural properties of the abstract transition functions $\Fo$ and $\Fu$ evaluated on some abstract states with finite ranking.

\begin{lemma}\label{lem:conditions on transitions}
	Suppose $\pi_0\in \Pi_0^\DM$ is some strategy of $\p{0}$.
	For every $v\in V_0$ with $r(v) = i \neq \infty$, $i > 0$, both  $\Fo(v,\pi_0(v))\cap r^{-1}(\infty) = \emptyset$ and either of the following holds:
	\begin{enumerate}
		\item $\Fu(v,\pi_0(v))\cap r^{-1}(i-1)\neq \emptyset$, or \label{item:intersection}
		\item $\Fu(v,\pi_0(v)) = \emptyset$ and $\Fo(v,\pi_0(v))\subseteq r^{-1}(i-1)$. \label{item:inclusion}
	\end{enumerate}
\end{lemma}

\begin{proof}
	Firstly, $\Fo(v,\pi_0(v))\cap r^{-1}(\infty) = \emptyset$ should always hold as otherwise $\p{1}$ would have a strategy to reach a state in $r^{-1}(\infty)$ with nonzero probability in the next step.
	
	Suppose \eqref{item:inclusion} does not hold, implying either (a) $\Fu(v,\pi_0(v)) \neq \emptyset$, or (b) the existence of a vertex $v'\in \Fo(v,\pi_0(v))$ with $r(v')\neq i-1$.
	Then $\Fu(v,\pi_0(v))\cap r^{-1}(i-1)\neq \emptyset$ must hold, as otherwise, for case (a) and (b) $\p{1}$ would have strategies $\pi_1$ with $\pi_1(v,\pi_0(v))=(\Fu(v,\pi_0(v)))$ and $\pi_1(v,\pi_0(v))=(\Fu(v,\pi_0(v))\cup \set{v'})$ respectively, such that $P_v^{\pi_0,\pi_1}(\game\models \bigcirc r^{-1}(i-1)) = 0$.
	
	On the other hand, suppose \eqref{item:intersection} does not hold.
	Then $\Fu(v,\pi_0(v)) = \emptyset$  must be true, as otherwise $\p{1}$ would have a strategy $\pi_1$ with $\pi_1(v,\pi_0(v))=(\Fu(v,\pi_0(v)))$ such that $P_v^{\pi_0,\pi_1}(\game\models \bigcirc r^{-1}(i-1)) = 0$.
	Moreover, $\Fo(v,\pi_0(v))\subseteq r^{-1}(i-1)$ must also be true, as otherwise there would exist a vertex $v'\in \Fo(v,\pi_0(v))$ with $r(v')\neq i-1$, and $\p{1}$ would have a strategy $\pi_1$ with $\pi_1(v,\pi_0(v))=(\Fu(v,\pi_0(v))\cup \set{v'}) = (\set{v'})$ such that $P_v^{\pi_0,\pi_1}(\game\models \bigcirc r^{-1}(i-1)) = 0$.
\end{proof}

The following lemma establishes soundness of the reduction with respect to reachability specifications.

\begin{proposition}\label{prop:sound reach}
	Let $\Sys$ be a CMP and $\game$ be a finite $2\sfrac{1}{2}$-player game graph as defined in \REFdef{def:reduced game}.
	Suppose there is a $\p{0}$ vertex $v\in V_0$ in $\game$ and a set of vertices $U\subseteq V_0$ for which there is a strategy $\pi_0\in \Pi_0^{\mathrm{DM}}$ of $\p{0}$ such that $\inf_{\pi_1\in \Pi_1} P_v^{\pi_0,\pi_1}(\game\models \lozenge U) = 1$.
	Then the refinement $\Cont\in \Policy$ of $\pi_0$ ensures that for every state $s\in v$, $P_s^\Cont(\Sys\models \lozenge Q^{-1}(U))=1$.
\end{proposition}

\begin{proof}
	It follows from the definition of the ranking function in \eqref{eq:ranking function} that the set of almost sure winning vertices for the specification $\lozenge U$ is given by all the vertices with finite rank.
	We show that for every vertex $v$ with a finite rank, the refinement $\Cont \in\Policy$ of $\pi_0$ ensures that from every state $\xs\in v$, $P_\xs^\Cont(\Sys\models \lozenge Q^{-1}(U))=1$.	
	
	First, trajectories starting from any state $\xs\in v$ with $r(v)\neq \infty$ never go to the region $Q^{-1}(r^{-1}(\infty))$. 
	This follows from the identity $\Fo(v,\pi_0(v))\cap r^{-1}(\infty) = \emptyset$ in \REFlem{lem:conditions on transitions} and because $\Fo(v,\pi_0(v))$ is an overapproximation of the one step reachable set from the states within vertex $v$.
	Hence, every infinite trajectory of $\Sys$ starting from $\xs$ will visit the states in $\Xs\setminus Q^{-1}(r^{-1}(\infty))$ infinitely often.
	
	The rest of the proof shows that if a trajectory visits the states $\Xs\setminus Q^{-1}(r^{-1}(\infty)\cup r^{-1}(0))$ infinitely often, then the trajectory will almost surely satisfy $\lozenge Q^{-1}(r^{-1}(0)) = \lozenge Q^{-1}(U)$.
	The is by induction over the largest rank assigned by $r$.
	For the base case, let the largest rank assigned by $r$ be $2$.
	We show that every state $\xs\in \Xs$ starting from inside a vertex $v$ with $r(v)=1$ or $r(v)=2$ will almost surely reach $Q^{-1}(U)$, i.e., $P_s^\Cont(\mathfrak{S} \models \square Q^{-1}(r^{-1}(1)\cup r^{-1}(2)))  = 0$. 
	Note that the events $\set{\lozenge\square Q^{-1}(r^{-1}(2))}$ and $\set{\square\lozenge Q^{-1}(r^{-1}(1))}$ form a partition of the event of $\set{ \square Q^{-1}(r^{-1}(1)\cup r^{-1}(2))}$. 
	Therefore,
	\begin{align*}
	P_s^\Cont &(\mathfrak{S} \models \square Q^{-1}(r^{-1}(1)\cup r^{-1}(2)))\\
	& = 
	P_s^\Cont(\mathfrak{S} \models \square Q^{-1}(r^{-1}(1)\cup r^{-1}(2))\wedge \lozenge\square Q^{-1}(r^{-1}(2)))\\
	&+P_s^\Cont(\mathfrak{S} \models \square Q^{-1}(r^{-1}(1)\cup r^{-1}(2))\wedge \square\lozenge Q^{-1}(r^{-1}(1))).
		\end{align*}
		The first term is upper bounded by $P_s^\Cont(\mathfrak{S} \models \lozenge\square Q^{-1}(r^{-1}(2)))$ which is zero, because $P_s^\Cont(\mathfrak{S} \models \square Q^{-1}(r^{-1}(2))) = \prod_{n=1}^\infty (1-\varepsilon)^n = 0$. 
		The second term is also zero because the event requires the number of transitions from $Q^{-1}(r^{-1}(1))$ to be infinite. 
		To see this, let $\boldsymbol i_n = (i_0,i_1,\ldots,i_n)$ be the first $(n+1)$ time instances that a trajectory visits $Q^{-1}(r^{-1}(1))$. 
		Then,
	\begin{align*}
	& P_s^\Cont(\mathfrak{S} \models \square Q^{-1}(r^{-1}(1)\cup r^{-1}(2))\wedge \square\lozenge Q^{-1}(r^{-1}(1))) =\\ 
	& \!\sum_{\boldsymbol i_n}\! P_s^\Cont(\mathfrak{S} \models \square Q^{-1}(r^{-1}(1)\!\cup\! r^{-1}(2))\!\wedge\! \square\lozenge Q^{-1}(r^{-1}(1))\,|\, \boldsymbol i_n)P_s^\Cont(\boldsymbol i_n)\\
	& \le \sum_{\boldsymbol i_n} (1-\varepsilon)^n P_s^\Cont(\boldsymbol i_n) = (1-\varepsilon)^n.
		\end{align*}
		The last inequality is due to either Cond.~\eqref{item:intersection} or Cond.~\eqref{item:inclusion} of \REFlem{lem:conditions on transitions} applied to the vertices in $v\in r^{-1}(1)$. 
		Note that this inequality holds for any $n$. 
		By taking the limit when $n$ goes to infinity, we have that this second term is also zero.	
	Hence the base case is established.
	
	For the induction hypothesis, assume that the claim holds when the maximum rank assigned by the function $r$ is $i$.
	Then for the induction step, i.e., when the maximum rank is $i+1$, we can follow same argument, as we did for the states with rank $2$ in the base case, to show that every infinite trajectory inside $\Xs\setminus Q^{-1}(r^{-1}(\infty)\cup r^{-1}(0))$ will never get trapped inside $Q^{-1}(r^{-1}(i+1))$, which will mean that the trajectory will visit the states in $\Xs\setminus Q^{-1}(r^{-1}(\infty)\cup r^{-1}(0) \cup r^{-1}(i+1))$ infinitely often.
	Then it follows from the induction hypothesis that the trajectories will reach $Q^{-1}(U)$ almost surely.
\end{proof}

\smallskip
\noindent{\textbf{Step 5: Refinement of runs.} 
We show that every finite path in $\Sys$ can be mapped to a positive probability finite run in $\game$; this will be used establish a bridge from the universal quantification over finite paths in $\Sys$ to the universal quantification over finite runs in $\game$.

\begin{lemma}\label{lem:existence of abstract trajectory}
	Let $\Sys$ be a CMP, $\game$ be the abstract game graph as defined in \REFdef{def:reduced game}, $\pi_0\in \Pi_0^\DM$ be an arbitrary $\p{0}$ strategy in the game $\game$, and $\xs\in \Xs$ be a state of $\Sys$.
	Suppose $\Cont\in\Policy$ is the refinement of $\pi_0$.
	Then for every finite trajectory $\xs^0\ldots \xs^n\in \Xs^*$ of $\Sys$ in the support of the distribution $P_{\xs^0}^\Cont$, there exists a $\p{1}$ strategy $\pi_1\in \Pi_1$ such that $P_{\xh^0}^{\pi_0,\pi_1}(\game\models \xh^0\ldots \xh^n) > 0$, where $\xh^i = Q(\xs^i)$ for every $i\in [0;n]$. 
\end{lemma}

\begin{proof}
	The initial state $\xs^0\in \xh^0$, and for every $0\leq i < n$, from the definition of $\Fo$ it follows that $\xh^{i+1}\in \Fo(\xh^i, \pi_0(\xh^i))$.
	Thus, from every $\p{1}$ vertex $(\xh^i,\pi_0(\xh^i))$, there is a successor vertex in $V_r$ whose successor is $\xh^{i+1}$.
	Hence, for every $0\leq i<n$, there is some move of $\p{1}$ which causes a transition to $\xh^{i+1}$ with some positive probability $p^i$.
	Then $P_{\xh^0}^{\pi_0,\pi_1}(\game\models \xh^0\ldots \xh^n) = \prod_{i=0}^{n-1} p^i > 0$. 
\end{proof}

\smallskip
\noindent{\textbf{Step 6: The final assembly of the proof.} 
Finally, we finish the proof of \REFthm{thm:main} by stitching everything together.
	It is known that memoryless strategies suffice for winning almost surely in $2\sfrac{1}{2}$-player parity games \citep{zielonka2004perfect}.
	Let $\pi_0^*\in \Pi^{\mathrm{DM}}_0$ be the witness strategy of $\p{0}$ to almost surely win from the vertex $\xh^\ast$ in the game $\tup{\game,\widehat{\mathcal{P}}}$, and $\Cont^\ast$ be the refinement of $\pi_0^*$.
	We claim that $\xh^\ast\subseteq \W(\Sys,\Cont^\ast)$.

	We will show that for \emph{every} finite path of $\Sys$ starting within $\xh^\ast$ and ending in some odd priority state $B_i$, eventually either $B_i$ will not be visited any more, or a state of higher even color will be visited.
	Then the claim will follow from \REFlem{lem:sufficient condition for parity in CMP}.
	We know from \REFlem{lem:existence of abstract trajectory} that existence of a finite path $\xs^0\ldots \xs^n$ of $\Sys$ implies existence of an abstract run $\xh^0\ldots \xh^n$ such that $\sup_{\pi_1\in \Pi_1}P_{\xh^0}^{\pi_0,\pi_1} (\game\models \xh^0\ldots\xh^n)>0$.
	Moreover, due to the priority preserving partitions we have $\xh^n\in \Bh_i$, where $i$ is odd.
	Since $\pi_0^*$ is an almost sure winning strategy, hence, by using \REFlem{lem:decompose proof of liveness}, we know that the following holds:
	\begin{align}\label{eq b}
		\inf_{\pi_1\in \Pi_1} P_{\xh^n}^{\pi_0^*,\pi_1} \left( \game\models \lozenge \left( \square \lnot \Bh_i \vee \cup_{j\ineven\inctwo{i+1}{k}} \Bh_j \right)\right) = 1.
	\end{align}
	From \REFprop{prop:sound safety}, we know that the set of abstract states from which the specification $\square \lnot \Bh_i$ is satisfied (almost) surely using the strategy $\pi_0^*$ are also the set of continuous states from which the specification $\square\lnot B_i$ is satisfied almost surely using the controller $\Cont^*$.
	Together with \REFprop{prop:sound reach} and \REFlem{lem:sufficient condition for parity in CMP}, we can infer \REFthm{thm:main} from \eqref{eq b}.

%% file: example.tex

\section{Numerical Example}
\label{sec:case_study}

We consider the controller synthesis problems for a two-dimensional stochastic bistable switch with a couple of different parity specifications; the examples have been adopted, mutatis mutandis, from the work of \citet{Coogan20}.
The dynamics of the system is modeled by the following difference equations:
\begin{align}
	&s_1^{k+1} = s_1^k + \left( -a\cdot s_1^k + s_2^k \right)\cdot \tau + u_1^k + \varsigma_1^k\\
	&s_2^{k+1} = s_2^k + \left( \frac{\left(s_1^k\right)^2}{\left( s_1^k\right)^2 +1} - b\cdot s_2^k \right)\cdot \tau + u_2^k + \varsigma_2^k,
\end{align}
where $s_1,s_2$ are the state variables, $u_1,u_2$ are the control inputs, $a,b$ are constant parameters, $\tau$ is the sampling time, and $\varsigma_1,\varsigma_2$ are stochastic noises.
We assume that the domain of the state variables is $[0.0,4.0]\times [0.0,4.0]$, and we saturate the state trajectories at the boundary of the domain.
We consider a finite set of values for both of the control inputs: for every $k$, $u_1^k,u_2^k\in \set{-0.05,0.0,0.05}$.
The values of the constants are given by: $a=1.3$, $b=0.25$, and $\tau = 0.05$.
Finally, we assume that the stochastic noise samples $(\varsigma_1^k,\varsigma_2^k)$ are drawn from a piecewise continuous density function with the support $D = [-0.4,-0.2]\times [-0.4,-0.2]$.\footnote{\citet{Coogan20} considered the density function to be given by truncated Gaussian distribution with support $D$.
In our work, we disregard the shape of the distribution because we restrict the focus to only the qualitative satisfaction of the specification.}

Let $A,B,C,D$ be sets of states, as shown in Fig.~\ref{fig:state predicates}.
We are interested in synthesizing the almost sure winning controllers for the above system for the following two LTL specifications:
\begin{align*}
	\varphi_1 &\coloneqq \square\left( \left( \lnot A\wedge \bigcirc A\right) \rightarrow \left( \bigcirc\bigcirc A\wedge \bigcirc\bigcirc\bigcirc A\right) \right),\\
	\varphi_2 &\coloneqq \left( \square\lozenge B\rightarrow \lozenge C\right) \wedge \left( \lozenge D\rightarrow \square\lnot C\right).
\end{align*}
The specifications $\varphi_1$ and $\varphi_2$ can be represented using the parity automata \footnote{\citet{Coogan20} modeled $\varphi_1$ and $\varphi_2$ using Rabin automata, and we transformed them into (language-) equivalent parity automata to match our setup.} shown in Fig.~\ref{fig:parity condition 1} and \ref{fig:parity condition 2}.
Firstly, for each of the two specifications, we compute a product with the system model.
\new{
The parity specification over the product state space can be represented in the same form as given in \eqref{eq:definition of parity}; we omit the details for succinctness.
}
Secondly, we apply the algorithm from Sec.~\ref{synthesis_game} on the product system to solve the synthesis problem.
Both of these steps have been implemented on the open-source tool \mascot \citep{majumdar2023flexible}.
By symbolically encoding the abstract $2\sfrac{1}{2}$-player game using BDD-s, and by using sophisticated acceleration techniques for solving symbolic fixpoint algorithms from the literature \citep{long1994improved,piterman2006faster}, we achieve significant improvement in performance, in comparison with the implementation of the enumerative algorithm of \citet{Coogan20} in the tool called \textsf{StochasticSynthesis}.\footnote{Repository: \texttt{https://github.com/gtfactslab/StochasticSynthesis}, commit nr.: \texttt{888b9dcf67369a732b8c225d790bd3343e4442e5}}
\new{
A further improvement of our implementation may be possible by using more modern and theoretically superior approaches to solving $\mu$-calculus expressions \citep{arnold2021quasi}; we leave this open for future research.
}

The tool \textsf{StochasticSynthesis} has a couple of additional features compared to our implementation in \mascot.
Firstly, it performs an adaptive abstraction refinement procedure for achieving better computational efficiency over uniform abstractions.
This is an orthogonal optimization which is known to be effective in discretization-based approaches for controller synthesis \citep{SA13, DBLP:journals/tac/GirardGM16,nilsson2017augmented,hsu2018multi}, and we expect that our uniform abstraction-based synthesis procedure will benefit further from this in the future.
Secondly, \textsf{StochasticSynthesis} also addresses the quantitative aspect of the synthesis problem, which is to maximize the probability of satisfying the given specification.
In all our experiments, we disabled this second feature of \textsf{StochasticSynthesis}, since the quantitative part is not the main algorithmic contribution of our paper.

\new{
On the other hand, \textsf{StochasticSynthesis} is written in Python whereas \mascot is written in C++.
Since C++ code is generally faster, some computational benefit is already expected due to the low-level implementation-specific differences.
However, we highlight that there is a clear complexity-theoretic advantage in using our approach, as discussed in Rem.~\ref{remark:complexity comparison with coogan}, which we expect to have contributed as well in the performance improvement of our tool.
}

We performed all the experiments on a Macbook Pro (2015) laptop equipped with a 2.7 GHz Intel Core i5 processor and a 16 GB RAM.

We performed two sets of experiments to compare the performance of \mascot and \textsf{StochasticSynthesis}, one with the adaptive refinement feature of \textsf{StochasticSynthesis} disabled and one with the same enabled.
The results have been summarized in Tab.~\ref{table:performance summary}, Fig.~\ref{fig:performance plot}, and Fig.~\ref{fig:visualization of winning regions}.
All the experiments empirically show that the approximation error reduces with finer discretization, as already mentioned in Rem.~\ref{remark:quality of approximation}.
We highlight the other main findings in the following:
(A) When the abstraction-refinement feature of \textsf{StochasticSynthesis} was disabled, \mascot outperformed \textsf{StochasticSynthesis} by a large margin.
In fact, for some levels of discretization, \textsf{StochasticSynthesis} crashed due to memory limitation, whereas \mascot consumed quite manageable amount of memory and synthesized controllers within reasonable amount of time.
(B) Even when the abstraction-refinement feature of \textsf{StochasticSynthesis} was enabled, for achieving the same level of approximation error, \mascot was significantly faster for $\varphi_1$ and was competitive for $\varphi_2$, and consumed much less memory for both $\varphi_1$ and $\varphi_2$:
 For $\varphi_1$, at one point \mascot was more than $150$ times faster and consumed around $150$ times lesser memory.
These findings demonstrate the superior capabilities of our symbolic solution approach, which can be potentially further improved by using adaptive refinement techniques.

\begin{figure}
	\centering
	\subfloat[][The state predicates.]{
	\begin{tikzpicture}[scale=0.8]
		\draw [draw=black!50!white] (0,0) grid  (4,4) rectangle (0,0);
		\draw [fill=red,opacity=0.3]	(1,1)	--	(3,1)	--	(3,3)	--	(2,3)	--	(2,2)	--	(1,2)	--	cycle;
		\draw [fill=green,opacity=0.3]	(0,0)	rectangle	(1,1);
		\draw [fill=blue,opacity=0.3]	(1,1)	rectangle	(2,2);
		\draw [fill=blue,opacity=0.3]	(3,1)	rectangle	(4,3);
		\draw [fill=blue,opacity=0.3]	(0,3)	rectangle	(1,4);
		\draw [fill=red!50!yellow,opacity=0.3]	(2,3)	rectangle	(3,4);

		\node at		(1.5,1.5)	{$A,C$};
		\node at		(2.5,1.5)	{$A$};
		\node at		(2.5,2.5)	{$A$};
		\node at		(0.5,0.5)	{$B$};
		\node at		(3.5,1.5)	{$C$};
		\node at		(3.5,2.5)	{$C$};
		\node at		(0.5,3.5)	{$C$};
		\node at		(2.5,3.5)	{$D$};
	\end{tikzpicture}
	\label{fig:state predicates}
	}
	\subfloat[][Parity automaton for $\varphi_1$; the prioritized partition is $\mathcal{P}=\tup{\emptyset,\set{q_2,q_3,q_4},\set{q_0,q_1}}$. $\lnot A$ stands for any element in $\set{\set{},B,C,D}$.]{
	\begin{tikzpicture}[scale=0.8]
		\tikzstyle{every node}=[font=\scriptsize]
		\node[state,initial]	(q0)	at	(0,0)	{$q_0$};
		\node[state]			(q1)	[right=of q0]	{$q_1$};
		\node[state]			(q2)	[below=of q1]	{$q_2$};
		\node[state]			(q3)	[below=of q0]	{$q_3$};
		\node[state]			(q4)	at	($(q0)!0.5!(q3)+(q0)!0.5!(q1)$)	{$q_4$};

		\draw[->]	(q0)	edge[loop above]	node	{$A$}		()
								edge						node[above] {$\lnot A$} (q1)
						(q1)	edge[loop above]	node {$\lnot A$}	()
								edge						node[right]	{$A$}		(q2)
						(q2)	edge						node[below]		{$\lnot A$}	(q3)
								edge						node[right,pos=0.8]		{$A$}			(q4)
						(q3)	edge[loop above]		node[left]		{$\lnot A,A$}	()
						(q4)	edge						node[below,pos=0.8]		{$A$}			(q0)
								edge						node[below,pos=0.1]		{$\lnot A$}	(q3);
	\end{tikzpicture}
	\label{fig:parity condition 1}
	}\\
	\subfloat[][Parity automaton for $\varphi_2$; the prioritized partition is $\mathcal{P}=\tup{\set{q_0,q_2,q_4},\set{q_1,q_3,q_5,q_6}}$.]{
	\begin{tikzpicture}[scale=0.5]
		\tikzstyle{every node}=[font=\scriptsize]
		\node[state,initial]	(q0)	at	(0,0)	{$q_0$};
		\node[state]			(q1)	[right=of q0]	{$q_1$};
		\node[state]			(q3)	at	($(q0)!0.5!(q1)+(0,-3)$)	{$q_3$};
		\node[state]			(q2)	[left=of q3]	{$q_2$};
		\node[state]			(q4)	[right=of q3]	{$q_4$};
		\node[state]			(q5)	at	($(q2)!0.5!(q3)+(0,-3)$)		{$q_5$};
		\node[state]			(q6)	at	($(q3)!0.5!(q4)+(0,-3)$)	{$q_6$};

		\draw[->]	(q0)	edge[loop above]	node	{$\set{},A$}	()
								edge[bend left]					node[above]	{$B$}		(q1)
								edge						node[pos=0.7,above]	{$C$}		(q4)
								edge						node[left]	{$D$}		(q2)
						(q1)	edge[loop above]	node	{$B$}			()
								edge				node[below]	{$\set{},A$}	(q0)
								edge						node	[right]			{$C$}	(q4)
								edge	[bend left=15]					node[above]				{$D$}	(q2)
						(q2)	edge	[loop left]		node				{$\set{},A,D$}	()
								edge					node	[below,pos=0.8]			{$B$}		(q3)
								edge						node	[left]			{$C$}	(q5)
						(q3)	edge	[bend left]					node[below]				{$\set{},A,D$}		(q2)
								edge [loop right]		node				{$B$}		()
								edge						node	[right]			{$C$}	(q5)
						(q4)	edge	[loop right]		node				{$\set{},A,B,C$}	()
								edge						node	[right]			{$D$}	(q6)
						(q5)	edge	[loop left]		node				{$\set{},A,B,C,D$}	()
						(q6)	edge	[loop right]		node				{$\set{},A,B,C,D$}	();
	\end{tikzpicture}
	\label{fig:parity condition 2}
	}
	\caption{The state predicates and the parity specifications.}
	\label{fig:example_spec}
\end{figure}

\begin{sidewaystable}
	\centering
	\begin{tabular}{
		|>{\centering\arraybackslash}m{1.5cm}|
		>{\centering\arraybackslash}m{2cm}|
		>{\centering\arraybackslash}m{1cm}|
		>{\centering\arraybackslash}m{1cm}|
		>{\centering\arraybackslash}m{1.4cm}|
		>{\centering\arraybackslash}m{1.4cm}|
		>{\centering\arraybackslash}m{1.4cm}|
		>{\centering\arraybackslash}m{1.4cm}|
		>{\centering\arraybackslash}m{1.4cm}|
		>{\centering\arraybackslash}m{1.4cm}|}
		\hline
		\multirow{3}{*}{\parbox{1.5cm}{\centering Specification}} & \multirow{3}{*}{\parbox{2cm}{\centering Size of \\ abstract states}} & \multicolumn{2}{c|}{Approximation error} & \multicolumn{2}{c|}{Abstraction time} & \multicolumn{2}{c|}{Total synthesis time} & \multicolumn{2}{c|}{Peak memory footprint}\\
		\cline{3-10}
		  & & {\multirow{2}{*}{\parbox{1cm}{\centering \mascot}}} & {\multirow{2}{*}{\parbox{1cm}{\centering \textsf{SS} }}}   & \multirow{2}{*}{\parbox{1.5cm}{\centering \mascot}} & \multirow{2}{*}{\parbox{1.5cm}{\centering \textsf{SS}}} & \multirow{2}{*}{\parbox{1.5cm}{\centering \mascot}} & \multirow{2}{*}{\parbox{1.5cm}{\centering \textsf{SS} }} & \multirow{2}{*}{\parbox{1.5cm}{\centering \mascot}} & \multirow{2}{*}{\parbox{1.5cm}{\centering \textsf{SS} }} \\
		  & & &  & & &  & & & \\
		\hline
		\multirow{5}{*}{$\varphi_1$} & $\sfrac{1}{2}\times \sfrac{1}{2}$ & \tablenum{7.0} & \tablenum{9.0}  & \num{0.003}\,\si{\second} & \num{0.4}\,\si{\second} & \num{0.02}\,\si{\second}  & \num{8}\,\si{\second} & \num{21}\,\si{\mebi\byte} & \num{125}\,\si{\mebi\byte} \\
		& $\sfrac{1}{4}\times \sfrac{1}{4}$ & \tablenum{6.6} & \tablenum{5.9}  & \num{0.04}\,\si{\second}  & \num{13}\,\si{\second} & \num{0.2}\,\si{\second}  & \num{18}\,\si{\second} & \num{23}\,\si{\mebi\byte} & \num{1}\,\si{\gibi\byte}  \\
		& $\sfrac{1}{8}\times \sfrac{1}{8}$ & \tablenum{4.0} & \tablenum{3.0}  & \num{0.2}\,\si{\second} & \num{35}\,\si{\minute}\,\num{5}\,\si{\second} &  \num{0.7}\,\si{\second} & \num{9}\,\si{\minute}\,\num{18}\,\si{\second} & \num{26}\,\si{\mebi\byte} & \num{81}\,\si{\gibi\byte} \\
		& $\sfrac{1}{16}\times \sfrac{1}{16}$  & \tablenum{1.7} & OoM & \num{0.9}\,\si{\second}  & OoM  & \num{5}\,\si{\second}  & OoM & \num{50}\,\si{\mebi\byte} & \num{127}\,\si{\gibi\byte} \\
		& $\sfrac{1}{32}\times \sfrac{1}{32}$ & \tablenum{0.8} & OoM & \num{5}\,\si{\second}  & OoM &  \num{37}\,\si{\second} & OoM & \num{171}\,\si{\mebi\byte} & \num{127}\,\si{\gibi\byte} \\
		 \hline
		 \multirow{5}{*}{$\varphi_2$} & $\sfrac{1}{2}\times \sfrac{1}{2}$ & \tablenum{11.0} & \tablenum{11.8}  & \num{0.004}\,\si{\second} & \num{0.6}\,\si{\second} & \num{5}\,\si{\second} & \num{30}\,\si{\second} & \num{864}\,\si{\mebi\byte} & \num{156}\,\si{\mebi\byte} \\
		& $\sfrac{1}{4}\times \sfrac{1}{4}$ & \tablenum{6.8}  & \tablenum{3.4}  & \num{0.02}\,\si{\second} & \num{15}\,\si{\second} & \num{13}\,\si{\second}  & \num{55}\,\si{\second} & \num{866}\,\si{\mebi\byte} & \num{1}\,\si{\gibi\byte} \\
		& $\sfrac{1}{8}\times \sfrac{1}{8}$ & \tablenum{2.0} & \tablenum{1.8} & \num{0.2}\,\si{\second} & \num{34}\,\si{\minute}\,\num{20}\,\si{\second} & \num{1}\,\si{\minute}\,\num{10}\,\si{\second} & \num{16}\,\si{\minute}\,\num{1}\,\si{\second} & \num{890}\,\si{\mebi\byte} & \num{81}\,\si{\gibi\byte} \\
		& $\sfrac{1}{16}\times \sfrac{1}{16}$ & \tablenum{1.1} & OoM & \num{0.9}\,\si{\second} & OoM & \num{4}\,\si{\minute}\,\num{37}\,\si{\second} & OoM & \num{994}\,\si{\mebi\byte} & \num{126}\,\si{\gibi\byte} \\
		& $\sfrac{1}{32}\times \sfrac{1}{32}$ & \tablenum{0.6} & OoM & \num{5}\,\si{\second}   & OoM  & \num{45}\,\si{\minute}\,\num{39}\,\si{\second} & OoM & \SI{1}{\gibi\byte} & \SI{127}{\gibi\byte} \\
		\hline
	\end{tabular}
	\caption{Performance comparison between \mascot and \textsf{StochasticSynthesis} (abbreviated as \textsf{SS}) \citep{Coogan20} when both tools are restricted to use \emph{uniform grid-based abstraction} only.
	Col.~$1$ shows the specification considered, Col.~$2$ shows the size of each individual abstract state, Col.~$3$ and $4$ compare the approximation errors, Col.~$5$ and $6$ compare the abstraction computation times, Col.~$7$ and $8$ compare the total synthesis times (combined time for computing the over- and the under-approximations of the almost sure winning regions), and Col.~$9$ and $10$ compare the peak memory footprint (as measured using the ``time'' command) for both tools.
	``OoM'' stands for out-of-memory.
	The length of the sides of the abstract states is measured in units, and the approximation error is measured in sq.\ units (the area covered by the abstract states which are in the over-approximation but not in the under-approximation).}
	\label{table:performance summary}
\end{sidewaystable}

\begin{figure}
	\centering
	\subfloat[][Specification $\varphi_1$]{
		\begin{tikzpicture}[scale=0.65]
			\begin{axis}[
						width=0.7\columnwidth,
						xlabel=Approximation error (in sq.\ units),
						ylabel=Total computation time (\si{\second}),
						xmin=0, ymin=0,
						ymode=log, log basis y={10}]
			\addplot table [y=tM_S1, x=VM_S1]{FIGURES/data_mascot.txt};
			\addlegendentry{\mascot}
			\addplot table [y=tSS_S1, x=VSS_S1]{FIGURES/data_ss.txt};
			\addlegendentry{\textsf{SS}}
			\end{axis}
		\end{tikzpicture}
		\begin{tikzpicture}[scale=0.65]
			\begin{axis}[
						width=0.7\columnwidth,
						xlabel=Approximation error (in sq.\ units),
						ylabel=Peak memory footprint (\si{\mebi\byte}),
						xmin=0, ymin=0,
						ymode=log, log basis y={10}]
			\addplot table	[y=mM_S1, x=VM_S1]{FIGURES/data_mascot.txt};
			\addlegendentry{\mascot}
			\addplot table	[y=mSS_S1, x=VSS_S1]{FIGURES/data_ss.txt};
			\addlegendentry{\textsf{SS}}
			\end{axis}
		\end{tikzpicture}
	}\\
	\subfloat[][Specification $\varphi_2$]{
		\begin{tikzpicture}[scale=0.65]
			\begin{axis}[
						width=0.7\columnwidth,
						xlabel=Approximation error (in sq.\ units),
						ylabel=Total computation time (\si{\second}),
						xmin=0, ymin=0,
						ymode=log, log basis y={10}]
			\addplot table [y=tM_S2, x=VM_S2]{FIGURES/data_mascot.txt};
			\addlegendentry{\mascot}
			\addplot table [y=tSS_S2, x=VSS_S2]{FIGURES/data_ss.txt};
			\addlegendentry{\textsf{SS}}
			\end{axis}
		\end{tikzpicture}
		\begin{tikzpicture}[scale=0.65]
			\begin{axis}[
						width=0.7\columnwidth,
						xlabel=Approximation error (in sq.\ units),
						ylabel=Peak memory footprint (\si{\mebi\byte}),
						xmin=0, ymin=0,
						ymode=log, log basis y={10}]
			\addplot table	[y=mM_S2, x=VM_S2]{FIGURES/data_mascot.txt};
			\addlegendentry{\mascot}
			\addplot table	[y=mSS_S2, x=VSS_S2]{FIGURES/data_ss.txt};
			\addlegendentry{\textsf{SS}}
			\end{axis}
		\end{tikzpicture}
	}
	\caption{Performance comparison between \mascot and \textsf{StochasticSynthesis} (abbreviated as \textsf{SS}) \citep{Coogan20} when the latter \emph{was allowed to use its inbuilt abstraction refinement process for better performance}.
	The different points on the plots were obtained by running \mascot using different sizes of abstract states, and running \textsf{SS} using different numbers of refinement stages.}
	\label{fig:performance plot}
\end{figure}

\begin{figure*}
	\centering
	\subfloat[][$\varphi_1$ using \textsf{StochasticSynthesis}]{
		\includegraphics[scale=0.35]{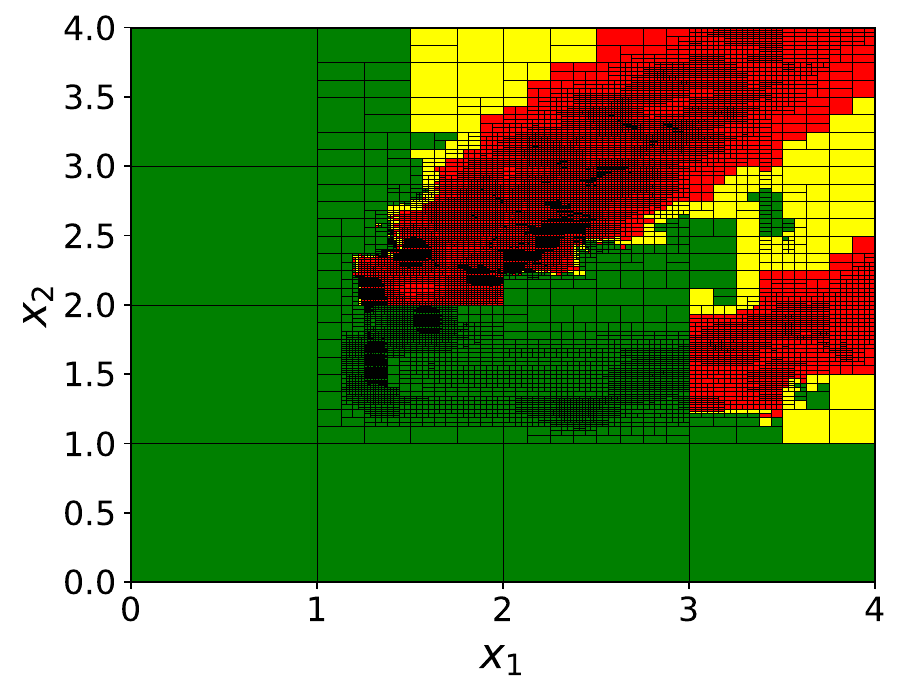}
	}
	\subfloat[][$\varphi_1$ using \mascot.]{
		\includegraphics[trim={0 6.75cm 0 0},scale=0.294]{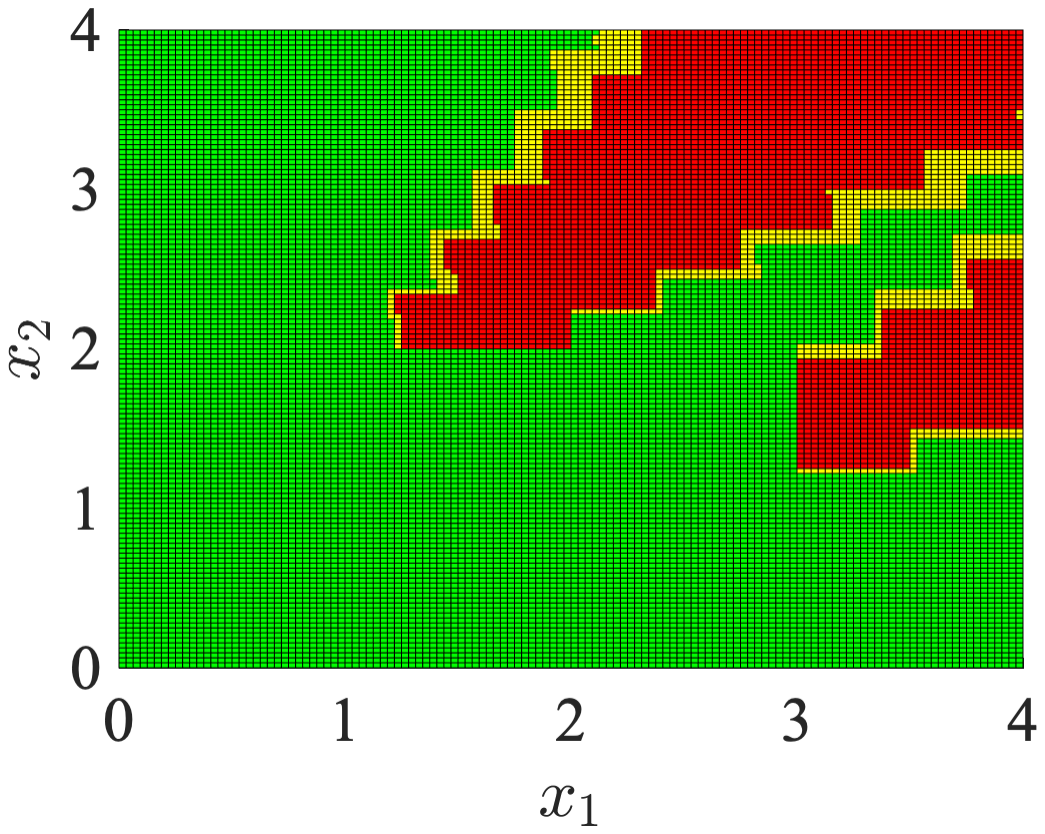}
	}\\
	\subfloat[][$\varphi_2$ using \textsf{StochasticSynthesis}.]{
		\includegraphics[scale=0.35]{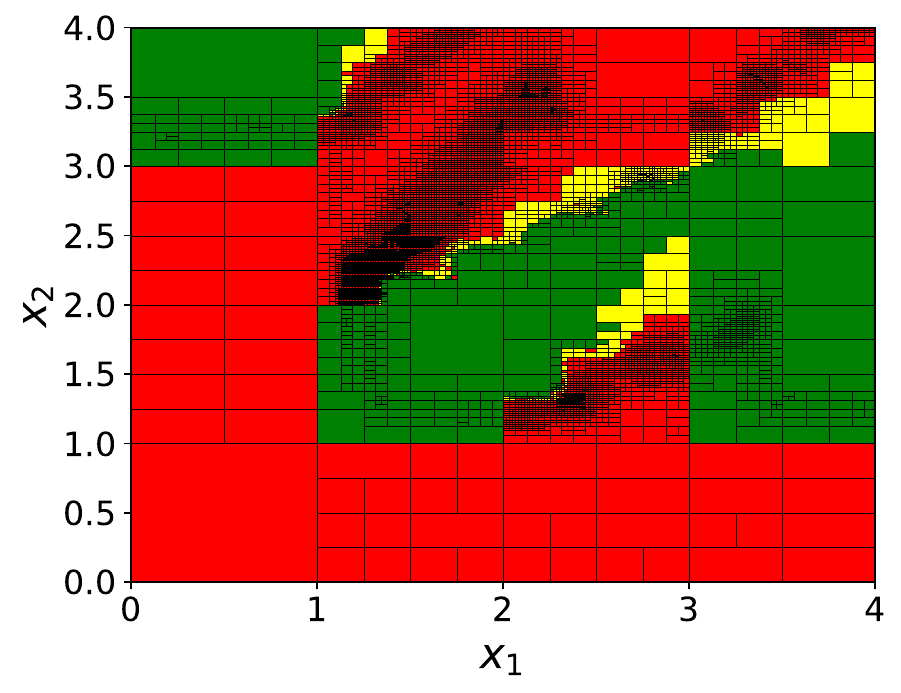}
	}
	\subfloat[][$\varphi_2$ using \mascot.]{
		\includegraphics[trim={0 6.75cm 0 0},scale=0.294]{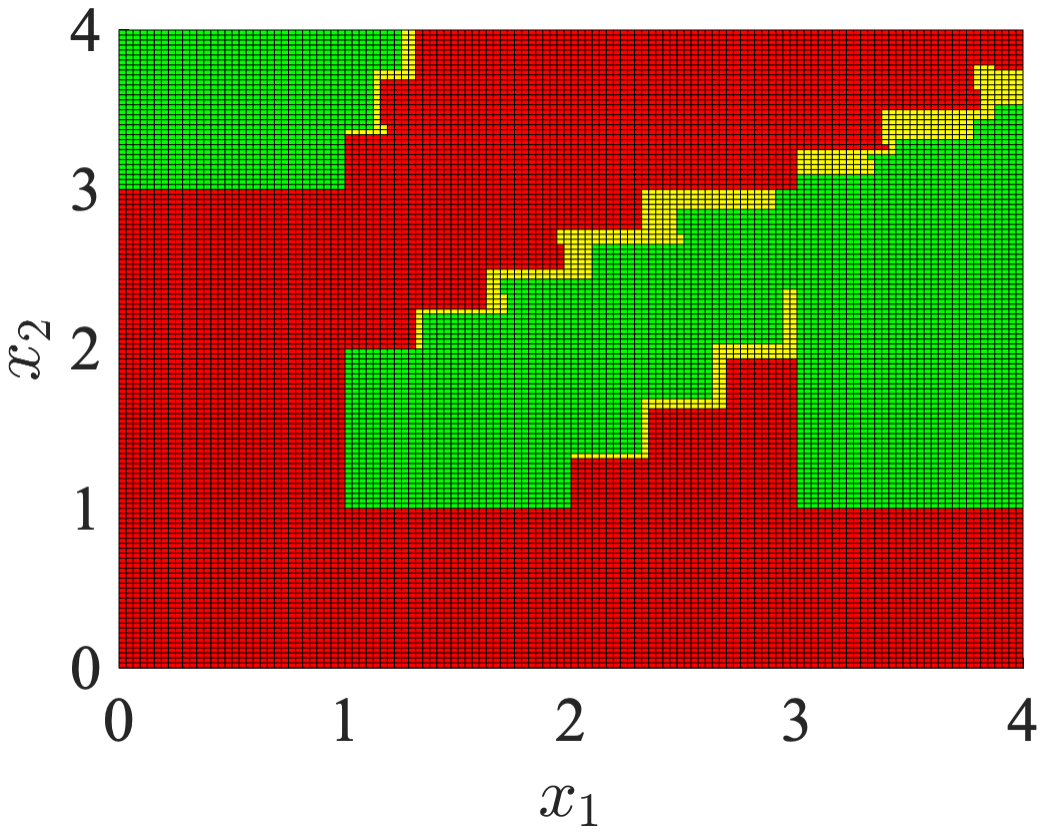}
	}
	\caption{Visualization of the under-approximations---in green---and the over-approximations---in yellow and green, combined---of the almost sure winning regions for the specifications $\varphi_1$ and $\varphi_2$ as computed using the tools \mascot and \textsf{StochasticSynthesis}; in red are the complements of the over-approximations.
	For \mascot, the abstract states were chosen of the size $\sfrac{1}{32}\times \sfrac{1}{32}$.
	For \textsf{StochasticSynthesis}, initially the abstracts states were chosen of the size $1 \times 1$, and then the tool was made to execute $14$ refinement steps to improve the approximation of the solution.}
	\label{fig:visualization of winning regions}
\end{figure*}


%% file: appendix.tex

\section{Appendix}

\subsection{Proof of Statements}

\begin{proof}[Proof of \REFthm{thm:the problem can be broken down}]
	By the definition of the winning set, we already know that $P_s^{\rho}(\mathfrak S\models \parity) = 1$ for all $s\in \W(\Sys,\rho)$.
	Take any $s\notin W\coloneqq \W(\Sys,\rho)$.
	Define $\tau$ to be the first time step when the path visits the set $W$.
	Note that $\tau$ is a random variable taking values in $\nat\cup\set{\infty}$.
	We use the law of total probability by making the event $(\mathfrak S\models \parity )$ conditional on $\tau$.
	Then we have
	\begin{align*}
	P_s^{\rho}& (\mathfrak S\models \parity)  \\
	&  = \sum_{n=0}^{\infty} P_s^{\rho}(\mathfrak S\models \parity\,|\, \tau=n)P_s^{\rho}(\tau=n)\\
	& + P_s^{\rho}(\mathfrak S\models \parity\,|\, \tau=\infty)P_s^{\rho}(\tau=\infty)\\
	& = \mathbb E_s^{\rho}\left[P_{s^n}^{\rho}(\mathfrak S\models \parity\,|\, s^1,s^2,\ldots,s^n, \tau=n)\right]\\
	& + P_s^{\rho}(\mathfrak S\models \parity \wedge \tau=\infty)\\
	& =^\ast \sum_{n=0}^{\infty} P_s^{\rho}(s^1,s^2,\ldots,s^{n-1}\in\mathcal S\backslash W, s^n\in W)\\
	& + P_s^{\rho}(\mathfrak S\models \parity \,\wedge\, \mathfrak S\models \square \mathcal S\backslash W)\\
        &\ge P_s^{\rho}(\mathfrak S\models \lozenge W).
	\end{align*}
	The equality (*) holds due to $s^n\!\in\! W$ and $P_{s^n}^{\rho}(\mathfrak S\!\models\! \parity)\!=\!1$.
\end{proof}

\begin{proof}[Proof of \REFthm{thm:absorbing}]
For any $s\in W$, we have
\begin{align*}
&P_s^{\rho}(\mathfrak S\models \parity)\\
&
= \int_{\mathcal S}P_{s_1}^{\rho}(\mathfrak S\models \parity)\Ker(ds_1|s,\rho(s))\\
&=  \int_{W} \!\!\Ker(ds_1|s,\rho(s)) + \int_{\mathcal S\backslash W} \!\!\!\!P_{s_1}^{\rho}(\mathfrak S\models \parity)\Ker(ds_1|s,\rho(s)).
\end{align*}
This means
\begin{align*}
&\int_{\mathcal S\backslash W} (1-P_{s_1}^{\rho}(\mathfrak S\models \parity))\Ker(ds_1|s,\rho(s))=0\Rightarrow\\
& \forall \epsilon>0,P_s^{\rho}\left[(1\!-\!P_{s_1}^{\rho}(\mathfrak S\models \parity))\mathbf{1}_{\mathcal S\backslash W}(s_1)\ge \epsilon\right]\!\le\! \frac{0}{\epsilon} \!=\! 0,
\end{align*}
where the last inequality is a consequence of Markov's inequality for non-negative random variables. By taking the union over a monotone positive sequence $\set{\epsilon_n\rightarrow 0}$, we get
\begin{align*}
&P_s^{\rho}\left[(1-P_{s_1}^{\rho}(\mathfrak S\models \parity))\mathbf{1}_{\mathcal S\backslash W}(s_1)>0 \right] = 0,\\
&P_s^{\rho}\left[s_1\in \mathcal S\backslash W \text{ and } P_{s_1}^{\rho}(\mathfrak S\models \parity)<1\right] = 0,\\
&P_s^{\rho}\left[ s_1\in \mathcal S\backslash W\right] = 0.
\end{align*}
\end{proof}

\begin{proof}[Proof of Thm.~\ref{thm:correctness of 2.5 parity fp}]
	Consider the following expression:
	\begin{align}\label{equ:2.5 parity or reach}
	I_\ell(T) \coloneqq
	& \nu Y_{\ell}.~\mu X_{\ell-1}\ldots ~\mu X_3.~\nu Y_{2}.~\mu X_1.\notag \\
	                &\quad(\Bh_1 \cap \apre(Y_2,X_1)) \cup (\Bh_2 \cap \cpre(Y_2)) \cup \notag\\
                &\quad\quad \left((\Bh_1\cup\Bh_2\cup\Bh_3)\cap \apre(Y_4, X_3)\right) \cup (\Bh_4\cap \cpre(Y_4)) \cup \notag\\
                &\quad\quad\quad\ldots\notag\\
                &\quad\quad\quad\quad \left((\Bh_1\cup\ldots\cup\Bh_{\ell-1})\cap \apre(Y_\ell,X_{\ell-1})\right)\cup  (\Bh_{\ell} \cap \cpre(Y_{\ell} ))\cup\notag\\
	                &\quad\quad\quad\quad\quad T.
	\end{align}
	We prove the stronger claim that the fixed point of $I_\ell(T)$ is the set of vertices from which $\p{0}$ can almost surely enforce the specification $\paritya\lor \lozenge T$.
	Then the claim will follow, by observing that the parity fixpoint in \eqref{equ:2.5 parity} for $\ell$ colors is syntactically equivalent to $I_\ell(\emptyset)$.

	We prove the claim by induction over the number of priories $\ell$.
	The base case of our inductive proof is the case $\widehat{\mathcal{P}}_2=\set{\Bh_1,\Bh_2}$, where
	\begin{align}\label{equ:2 color 2.5 parity or reach}
		I_2(T) = \nu Y_2.~\mu X_1.~(\Bh_1\cap\apre(Y_2,X_1)) \cup (\Bh_2\cap \cpre(Y_2)) \cup T,
	\end{align}
	which is \emph{exactly} equal to the almost sure winning region for the winning condition $\varphi \coloneqq \square\lozenge\Bh_2 \lor  \lozenge T$ (by Lem.~\ref{lem:a.s. safe buchi or reach} with the observation that $\Bh_1 = \Bh_2^c$ since the priorities partition the set of vertices).
	The base case is established by observing that $\varphi=\mathit{Parity}(\widehat{\mathcal{P}}_2) \lor  \lozenge T$.

	As the induction hypothesis, suppose when the set of priorities is $\widehat{\mathcal{P}}_{\ell-2}$, then the fixed point of $I_{\ell-2}(T)$ is exactly equal to the almost sure winning region for the specification $\mathit{Parity}(\widehat{\mathcal{P}}_{\ell-2})\lor  \lozenge T$.
	We show that the claim will also hold when the set of priorites is $\widehat{\mathcal{P}}_{\ell}$.
	Our proof of the induction claim uses the observation that almost sure satisfaction of $\mathit{Parity}(\widehat{\mathcal{P}}_{\ell})\lor  \lozenge T$ is achieved by either almost surely reaching the almost sure winning region for $\mathit{Parity}(\mathcal{P}_{\ell-2})$ or almost surely reaching a BSCC where $\Bh_{\ell}$ or $T$ can be repeatedly visited almost surely.

	Suppose $\game=\tup{V,E,\tup{V_0,V_1,V_r}}$ is the game graph and $\widehat{\mathcal{P}}_{\ell}$ is the set of priorities.
	Observe that in $\game$, the almost sure winning region for $\mathit{Parity}(\widehat{\mathcal{P}}_{\ell-2})$ is not well-defined, because the priorities $\widehat{\mathcal{P}}_{\ell-2}$ no longer create a partition of the set of vertices as required by the parity condition (assuming $\Bh_{\ell-1}$ and $\Bh_\ell$ to be nonempty).
	For this subtle technical reason, we need to consider a subgame constructed by removing the $\Bh_{\ell-1}$ and $\Bh_\ell$ vertices from the vertices in $\game$.
	Formally, we will consider the subgame $\widetilde{\game}=\tup{\widetilde{V},\widetilde{E},\tup{\widetilde{V}_0,\widetilde{V}_1,\widetilde{V}_r}}$, where $\widetilde{V}\subset V$ is obtained by replacing all the $\Bh_{\ell}$ and $\Bh_{\ell-1}$ vertices in $V$ by a special sink vertex $v_{\mathit{sink}}$ in $\widetilde{V}$, and by setting
	$\widetilde{E} \coloneqq (E \cap \widetilde{V}\times\widetilde{V}) \cup \set{(v,v_{\mathit{sink}}) \mid v\in \widetilde{V} \land \exists v'\in \Bh_{\ell-1}\cup \Bh_\ell\;.\;(v,v')\in E} \cup \set{(v_{\mathit{sink}},v_{\mathit{sink}})}$,
	$\widetilde{V}_0 \coloneqq (V_0 \cap \widetilde{V}) \cup v_{\mathit{sink}}$,
	$\widetilde{V}_1 \coloneqq V_1\cap \widetilde{V}$, and
	$\widetilde{V}_r \coloneqq V_r \cap \widetilde{V}$.
	Let $v_{sink} \in \Bh_{\ell-3}$ (any odd priority will work).
	We define the almost sure winning region for the specification $\mathit{Parity}(\mathcal{\widehat{P}}_{\ell-2})$ in the game graph $\game$ as the almost sure winning region for the same specification in the subgame graph $\widetilde{\game}$.

	Now we rewrite $I_{\ell}(T)$ as in the following:
	\begin{align}
		I_{\ell}(T) =& \notag\\
			&\nu Y_{\ell}.~\mu X_{\ell-1}.~\notag\\
			&\quad I_{\ell-2}\left(\left((\Bh_1\cup\ldots\cup\Bh_{\ell-1})\cap \apre(Y_{\ell},X_{\ell-1})\right) \cup (\Bh_{\ell}\cap \cpre(Y_{\ell}))\right) \cup T.\label{eq:parity l+2 color as a function of l color}
	\end{align}

	We establsih the soundness and the completeness claims separately:
	\begin{description}
		\item[Soundness:] Let $Y_{\ell}^*$ be the fixed point of the outer $\mu$-calculus iterations over $Y_{\ell}$ in the expression \eqref{eq:parity l+2 color as a function of l color}.
		For the last outer iteration with $Y_{\ell} = Y_{\ell}^*$, we obtain a growing sequence of $X_{\ell-1}$-s which represent the almost sure winning regions of the following specifications (applying the definition of $I_{\ell-2}(\cdot)$):
		\begin{align}
			X_{\ell-1}^0 &: \emptyset\notag\\
			X_{\ell-1}^1 &: \mathit{Parity}(\widehat{\mathcal{P}}_{\ell-2}) \lor \lozenge \left(\Bh_{\ell}\cap \cpre(Y_{\ell}^*)\right) \lor T\notag\\
			X_{\ell-1}^2 &: \mathit{Parity}(\widehat{\mathcal{P}}_{\ell-2}) \lor  \lozenge\left(\left((\Bh_1\cup\ldots\cup\Bh_{\ell-1})\cap \apre(Y_{\ell}^*,X_{\ell-1}^{1})\right) \cup (\Bh_{\ell}\cap \cpre(Y_{\ell}^*))\right) \lor T\notag\\
			\vdots\notag\\
			X_{\ell-1}^p &: \mathit{Parity}(\widehat{\mathcal{P}}_{\ell-2}) \lor \lozenge \left(\left((\Bh_1\cup\ldots\cup\Bh_{\ell-1})\cap \apre(Y_{\ell}^*,X_{\ell-1}^{p-1})\right) \cup (\Bh_{\ell}\cap \cpre(Y_{\ell}^*))\right) \lor T, \label{eq:parity outermost x sequence}
		\end{align}
		where $p\geq 1$ is the index at which the fixpoint over $X_{\ell-1}$ converges, i.e.\ $X_{\ell-1}^p = X_{\ell-1}^{p+1}$.
		We use a similar argument as in Lem.~\ref{lem:a.s. safe buchi or reach}:
		First, let us assume that all the vertices in the almost sure winning region of $\mathit{Parity}(\widehat{P}_{\ell-2})$ and the set $T$ are transformed into sink states with self-loops.
		(We do not care anymore once the runs reach these sets.)
		Consider any vertex $v\in Y_{\ell}^*$.
		From the sequence \eqref{eq:parity outermost x sequence}, we can infer that there is a $\p{0}$ strategy $\pi_0$ such that for every strategy $\pi_1$ of $\p{1}$, there is a run of positive probability of length $\leq p$ that \emph{either} (a) reaches $\Bh_{\ell}\cap \cpre(Y_{\ell}^*)$ \emph{or} (b) reaches the almost sure winning region for $\mathit{Parity}(\widehat{\mathcal{P}}_{\ell-2})$ \emph{or} (c) reaches $ T$.
		Moreover, $\pi_0$ ensures that no run from $v$ goes outside $Y_{\ell}^*$ even after reaching $\Bh_{\ell}$, so that (a) can be repeated or (b) or (c) can be eventually executed with positive probability.
		These imply that every path in the Markov chain, induced on $\game$ by $\pi_0$ and any arbitrary $\p{1}$ strategy $\pi_1$, will eventually reach a BSCC that \emph{either} (a) intersects with $\Bh_{\ell}$ \emph{or} (b) is in the almost sure winning region for $\mathit{Parity}(\widehat{\mathcal{P}}_{\ell-2})$ \emph{or} (c) is in $ T$.
		Then the induction claim follows from the ergodicity property of Markov chains (see \citet[Prop.~4.27, 4.28]{kemenydenumerable}).
		\item[Completeness:]
		Let $Y_{\ell}^0=V\supseteq Y_{\ell}^1\supseteq Y_{\ell}^2\supseteq \ldots \supseteq Y_{\ell}^q$ be the sequence of sets of vertices represented by the $Y_{\ell}$ fixpoint variable, where $q$ is the iteration index where the fixpoint converges (i.e.\ the minimum index $q$ such that $Y_{\ell}^q=Y_{\ell}^{q+1}$).
		We show, by induction, that for every iteration index $i\in [0;q]$, the almost sure winning region for the parity specification is contained in $Y_{\ell}^i$.

		The base case is when $i=0$, in which case the claim follows trivially by observing that $Y_{\ell}^0=V$.
		Suppose the claim holds for iteration $i$, i.e.\ $Y_{\ell}^i$ contains the almost sure winning region and every vertex outside $Y_{\ell}^i$ is losing with positive probability.
		When the iterations over $Y_{\ell}$ reaches $i+1$, for every vertex $v$ in $Y_{\ell}^{i+1}\setminus Y_{\ell}^{i}$, for every $\p{0}$ strategy there is a $\p{1}$ strategy such that
		either (a) the plays eventually exit $Y_{\ell}^i$, surely or with positive probability, or (b) no positive probability path exists from $v$ that eventually reaches $\Bh_{\ell}\cap \cpre(Y_{\ell}^i)$ or $T$ or the almost sure winning region for $\mathit{Parity}(\widehat{\mathcal{P}}_{\ell-2})$.
		(Both follow from the definition of $\cpre$ and $\apre$ and the fixpoint expression.)
		Both (a) and (b) result in failure to satisfy the parity winning condition $\mathit{Parity}(\widehat{\mathcal{P}}_{\ell})$, which imply that every vertex $v$ in $Y_{\ell}^{i+1}\setminus Y_{\ell}^{i}$ does not belong to the almost sure winning region of $\mathit{Parity}(\widehat{\mathcal{P}}_{\ell})$.
		This proves our induction claim.
		\new{
		\item[Time-complexity.]
			The time-complexity bound follows from the work by \cite{arnold2021quasi}, by observing that the number of bits to represent every $X_i$ or $Y_i$ set is $n$, and moreover every call to the innermost $\mu,\nu$-operator-free part of the formula requires $\mathcal{O}(\ell)$ number of symbolic steps (computing the predecessors and performing the set operations).
		\item[Space-complexity.]
			It follows from the work by	\cite{arnold2021quasi} that we will need $\mathcal{O}(n\cdot \polylog(n))$ space, for bookkeeping the intermediate results of the fixpoint variables.
			Computation of the $\mu,\nu$-operator-free part of the formula can be done in the same $\mathcal{O}(n)$ space, thus avoiding any additional blow-up.
			}
	\end{description}

\end{proof}

%% file: reportbib.bib
@article{althoff2021set,
  title={Set propagation techniques for reachability analysis},
  author={Althoff, Matthias and Frehse, Goran and Girard, Antoine},
  journal={Annual Review of Control, Robotics, and Autonomous Systems},
  volume={4},
  pages={369--395},
  year={2021},
  publisher={Annual Reviews}
}

@INPROCEEDINGS{MESG13,
	AUTHOR = {Maryam {Kamgarpour et al.}},
	TITLE = {Modeling Options for Demand Side Participation of Thermostatically Controlled Loads},
	BOOKTITLE = {International Conference on Bulk Power System Dynamics and Control (IREP)},
	MONTH = {August},
	YEAR = {2013},
	PAGES = {1--15},
}

@article{rawlings2017model,
  title={Model predictive control with discrete actuators: Theory and application},
  author={Rawlings, James B and Risbeck, Michael J},
  journal={Automatica},
  volume={78},
  pages={258--265},
  year={2017},
  publisher={Elsevier}
}

@article{wei2013bang,
  title={Bang-bang control applied in airfoil roll control with plasma actuators},
  author={Wei, Qingkai and Niu, Zhongguo and Chen, Bao and Huang, Xun},
  journal={Journal of aircraft},
  volume={50},
  number={2},
  pages={670--677},
  year={2013},
  publisher={American Institute of Aeronautics and Astronautics}
}

@inproceedings{arnold2021quasi,
  title={A quasi-polynomial black-box algorithm for fixed point evaluation},
  author={Arnold, Andr{\'e} and Niwi{\'n}ski, Damian and Parys, Pawe{\l}},
  booktitle={29th EACSL Annual Conference on Computer Science Logic (CSL 2021)},
  year={2021},
  organization={Schloss Dagstuhl-Leibniz-Zentrum f{\"u}r Informatik}
}

@article{CHSurvey,
  author    = {Krishnendu Chatterjee and
               Thomas A. Henzinger},
  title     = {A survey of stochastic {\(\omega\)}-regular games},
  journal   = {J. Comput. Syst. Sci.},
  volume    = {78},
  number    = {2},
  pages     = {394--413},
  year      = {2012},
}

@inproceedings{quigley2011low,
  title={A low-cost compliant 7-DOF robotic manipulator},
  author={Quigley, Morgan and Asbeck, Alan and Ng, Andrew},
  booktitle={2011 IEEE International Conference on Robotics and Automation},
  pages={6051--6058},
  year={2011},
  organization={IEEE}
}

@article{kemenydenumerable,
  title={Denumerable {M}arkov chains. With a chapter on {M}arkov random fields, by {D}avid {G}riffeath},
  author={Kemeny, John G and Snell, J Laurie and Knapp, Anthony W},
  journal={Graduate Texts in Mathematics},
  number={40},
  pages={0348--60090},
  publisher={Springer},
  year = {1976},
}

@article{majumdar2021adhs,
  title={Symbolic qualitative control for stochastic systems via finite parity games},
  author={Majumdar, Rupak and Mallik, Kaushik and Schmuck, Anne-Kathrin and Soudjani, Sadegh},
  journal={IFAC-PapersOnLine},
  volume={54},
  number={5},
  pages={127--132},
  year={2021},
  publisher={Elsevier}
}

@article{banerjee2023fast,
  title={Fast symbolic algorithms for omega-regular games under strong transition fairness},
  author={Banerjee, Tamajit and Majumdar, Rupak and Mallik, Kaushik and Schmuck, Anne-Kathrin and Soudjani, Sadegh},
  journal={TheoretiCS},
  volume={2},
  year={2023},
  publisher={Episciences. org}
}

@inproceedings{banerjee2022tacas,
  title={A direct symbolic algorithm for solving stochastic {R}abin games},
  author={Banerjee, Tamajit and Majumdar, Rupak and Mallik, Kaushik and Schmuck, Anne-Kathrin and Soudjani, Sadegh},
  booktitle={Tools and Algorithms for the Construction and Analysis of Systems, TACAS'22},
  pages={81--98},
  year={2022},
  organization={Springer}
}

@inproceedings{majumdar2023flexible,
  title={A Flexible Toolchain for Symbolic {R}abin Games under Fair and Stochastic Uncertainties},
  author={Majumdar, Rupak and Mallik, Kaushik and Rychlicki, Mateusz and Schmuck, Anne-Kathrin and Soudjani, Sadegh},
  booktitle={Computer Aided Verification: 35th International Conference, CAV 2023},
  note={(to appear)},
  year={2023},
  organization={Springer}
}

@book{tabuada09,
  title={Verification and Control of Hybrid Systems: A Symbolic Approach},
  author={Tabuada, P.},
  isbn={9781441902245},
  lccn={2009926478},
  year={2009},
  publisher={Springer}
}

@InProceedings{Prism-games,
author={M. Kwiatkowska and G. Norman and D. Parker and G. Santos},
title={{PRISM}-games 3.0: Stochastic Game Verification with Concurrency, Equilibria and Time},
booktitle={Proc. 32nd International Conference on Computer Aided Verification (CAV'20)},
year={2020},
pages={475--487},
publisher={Springer},
series={LNCS},
volume={12225},
}

@InProceedings{CourcoubetisYannakakis,
author="Courcoubetis, Costas
and Yannakakis, Mihalis",
title="{M}arkov decision processes and regular events",
booktitle="Automata, Languages and Programming",
year="1990",
publisher="Springer Berlin Heidelberg",
address="Berlin, Heidelberg",
pages="336--349",
}

@inproceedings{Vardi85,
  title={Automatic verification of probabilistic concurrent finite state programs},
  author={Vardi, Moshe Y},
  booktitle={26th Annual Symposium on Foundations of Computer Science},
  pages={327--338},
  year={1985},
  organization={IEEE}
}

@article{althoff2013reachability,
  title={Reachability analysis of nonlinear differential-algebraic systems},
  author={Althoff, Matthias and Krogh, Bruce H},
  journal={IEEE Transactions on Automatic Control},
  volume={59},
  number={2},
  pages={371--383},
  year={2013},
  publisher={IEEE}
}

@inproceedings{zielonka2004perfect,
  title={Perfect-information stochastic parity games},
  author={Zielonka, Wies{\l}aw},
  booktitle={International Conference on Foundations of Software Science and Computation Structures},
  pages={499--513},
  year={2004},
  organization={Springer}
}

@book{gradel2002automata,
  title={Automata, logics, and infinite games: a guide to current research},
  author={Gradel, Erich and Thomas, Wolfgang},
  volume={2500},
  year={2002},
  publisher={Springer Science \& Business Media}
}

@BOOK{hll1996,
	TITLE = {Discrete-time {M}arkov control processes},
	PUBLISHER = {Springer},
	YEAR = {1996},
	AUTHOR = {Hern{\'a}ndez-Lerma, O. and Lasserre, J. B.},
	VOLUME = {30},
	SERIES = {Applications of Mathematics},
}

@article{Kariotoglou17,
 author = {Kariotoglou, Nikolaos and Kamgarpour, Maryam and Summers, Tyler H. and Lygeros, John},
 title = {The Linear Programming Approach to Reach-avoid Problems for {M}arkov Decision Processes},
 journal = {J. Artif. Int. Res.},
 volume = {60},
 number = {1},
 month = sep,
 year = {2017},
 issn = {1076-9757},
 pages = {263--285},
 numpages = {23},
 publisher = {AI Access Foundation},
 address = {USA},
}

@article{Belta17,
  title={Temporal logic control for stochastic linear systems using abstraction refinement of probabilistic games},
  author={Svore{\v{n}}ov{\'a}, M{\'a}ria and K{\v{r}}et{\'\i}nsk{\`y}, Jan and Chmel{\'\i}k, Martin and Chatterjee, Krishnendu and {\v{C}}ern{\'a}, Ivana and Belta, Calin},
  journal={Nonlinear Analysis: Hybrid Systems},
  volume={23},
  pages={230--253},
  year={2017},
  publisher={Elsevier}
}

@ARTICLE{LO17, 
author={K. {Lesser} and M. {Oishi}}, 
journal={IEEE Transactions on Automatic Control}, 
title={Approximate Safety Verification and Control of Partially Observable Stochastic Hybrid Systems}, 
year={2017}, 
volume={62}, 
number={1}, 
pages={81-96}, 
month={Jan},}

@inproceedings{Vinod18,
 author = {Vinod, Abraham P. and Oishi, Meeko M. K.},
 title = {Scalable Underapproximative Verification of Stochastic {LTI} Systems Using Convexity and Compactness},
 booktitle = {Proceedings of the 21st International Conference on Hybrid Systems: Computation and Control (Part of CPS Week)},
 series = {HSCC '18},
 year = {2018},
 isbn = {978-1-4503-5642-8},
 location = {Porto, Portugal},
 pages = {1--10},
 numpages = {10},
 acmid = {3178148},
 publisher = {ACM},
 address = {New York, NY, USA},
}

@inproceedings{CJH03,
  author    = {Krishnendu Chatterjee and
               Marcin Jurdzinski and
               Thomas A. Henzinger},
  title     = {Simple Stochastic Parity Games},
  booktitle = {Computer Science Logic, 17th International Workshop, {CSL} 2003, 12th
               Annual Conference of the EACSL, and 8th Kurt G{\"{o}}del Colloquium,
               {KGC} 2003, Vienna, Austria, August 25-30, 2003, Proceedings},
  series    = {Lecture Notes in Computer Science},
  volume    = {2803},
  pages     = {100--113},
  publisher = {Springer},
  year      = {2003},
}

@article{LAVAE19,
title = "Compositional construction of infinite abstractions for networks of stochastic control systems",
journal = "Automatica",
volume = "107",
pages = "125 - 137",
year = "2019",
issn = "0005-1098",
author = "Abolfazl Lavaei and Sadegh Soudjani and Majid Zamani",
}

@inproceedings{hsu2018multi,
  title={Multi-layered abstraction-based controller synthesis for continuous-time systems},
  author={Hsu, Kyle and Majumdar, Rupak and Mallik, Kaushik and Schmuck, Anne-Kathrin},
  booktitle={Proceedings of the 21st International Conference on Hybrid Systems: Computation and Control (part of CPS Week)},
  pages={120--129},
  year={2018}
}

@article{Jagtap2019,
  title={Formal synthesis of stochastic systems via control barrier certificates},
  author={Jagtap, Pushpak and Soudjani, Sadegh and Zamani, Majid},
  journal={IEEE Transactions on Automatic Control},
  volume={66},
  number={7},
  pages={3097--3110},
  year={2020},
  publisher={IEEE}
}

@Article{SAM17,
author="Soudjani, Sadegh
and Abate, Alessandro
and Majumdar, Rupak",
title="Dynamic {B}ayesian networks for formal verification of structured stochastic processes",
journal="Acta Informatica",
year="2017",
month="Mar",
day="01",
volume="54",
number="2",
pages="217--242",
}

@ARTICLE{LAB15, 
author={M. Lahijanian and S. B. Andersson and C. Belta}, 
journal={IEEE Transactions on Automatic Control}, 
title={Formal Verification and Synthesis for Discrete-Time Stochastic Systems}, 
year={2015}, 
volume={60}, 
number={8}, 
pages={2031-2045}, 
month={Aug},}

@article{Kozen83,
title = "Results on the propositional $\mu$-calculus",
journal = "Theoretical Computer Science",
volume = "27",
number = "3",
pages = "333 - 354",
year = "1983",
note = "Special Issue Ninth International Colloquium on Automata, Languages and Programming (ICALP) Aarhus, Summer 1982",
issn = "0304-3975",
author = "Dexter Kozen",
}

@inproceedings{majumdar2020symbolic,
  title={Symbolic controller synthesis for {B}{\"u}chi specifications on stochastic systems},
  author={Majumdar, Rupak and Mallik, Kaushik and Soudjani, Sadegh},
  booktitle={Proceedings of the 23rd International Conference on Hybrid Systems: Computation and Control},
  pages={1--11},
  year={2020}
}

@article{nilsson2017augmented,
  title={Augmented finite transition systems as abstractions for control synthesis},
  author={Nilsson, Petter and Ozay, Necmiye and Liu, Jun},
  journal={Discrete Event Dynamic Systems},
  volume={27},
  number={2},
  pages={301--340},
  year={2017},
  publisher={Springer}
}

@article{Pnueli83,
 author = {Hart, Sergiu and Sharir, Micha and Pnueli, Amir},
 title = {Termination of Probabilistic Concurrent Program},
 journal = {ACM Trans. Program. Lang. Syst.},
 issue_date = {July 1983},
 volume = {5},
 number = {3},
 month = jul,
 year = {1983},
 pages = {356--380},
 numpages = {25},
 publisher = {ACM},
 address = {New York, NY, USA},
}

@inproceedings{coogan2015efficient,
  title={Efficient finite abstraction of mixed monotone systems},
  author={Coogan, Samuel and Arcak, Murat},
  booktitle={Proceedings of the 18th International Conference on Hybrid Systems: Computation and Control},
  pages={58--67},
  year={2015},
  organization={ACM},
}

@article{Coogan19,
  title={Specification-guided verification and abstraction refinement of mixed monotone stochastic systems},
  author={Dutreix, Maxence and Coogan, Samuel},
  journal={IEEE Transactions on Automatic Control},
  volume={66},
  number={7},
  pages={2975--2990},
  year={2020},
  publisher={IEEE}
}

@InProceedings{BdA95,
author="Bianco, Andrea
and de Alfaro, Luca",
editor="Thiagarajan, P. S.",
title="Model checking of probabilistic and nondeterministic systems",
booktitle="Foundations of Software Technology and Theoretical Computer Science",
year="1995",
publisher="Springer Berlin Heidelberg",
address="Berlin, Heidelberg",
pages="499--513",
}

@inproceedings{dAH00,
  title={Concurrent omega-regular games},
  author={de Alfaro, Luca and Henzinger, Thomas A},
  booktitle={Proceedings Fifteenth Annual IEEE Symposium on Logic in Computer Science (Cat. No. 99CB36332)},
  pages={141--154},
  year={2000},
  organization={IEEE}
}

@article{DBLP:journals/tac/GirardGM16,
  author    = {Antoine Girard and
               Gregor G{\"{o}}{\ss}ler and
               Sebti Mouelhi},
  title     = {Safety Controller Synthesis for Incrementally Stable Switched Systems
               Using Multiscale Symbolic Models},
  journal   = {{IEEE} Trans. Autom. Control.},
  volume    = {61},
  number    = {6},
  pages     = {1537--1549},
  year      = {2016},
  }

@ARTICLE{ReissigWeberRungger_2017_FRR,
author={G. Reissig and A. Weber and M. Rungger},
journal={{TAC}},
title={Feedback Refinement Relations for the Synthesis of Symbolic Controllers},
year={2017},
volume={62},
number={4},
pages={1781-1796},}

@BOOK{k2002,
	TITLE = {Foundations of modern probability},
	PUBLISHER = {Springer Verlag},
	YEAR = {2002},
	AUTHOR = {O. Kallenberg},
	SERIES = {Probability and its Applications},
	ADDRESS = {New York},
}

@inproceedings{TA12,
	AUTHOR = {I. Tkachev and A. Abate},
	TITLE = {Regularization of {B}ellman equations for infinite-horizon probabilistic properties},
	BOOKTITLE = {Proceedings of the 15th {ACM} international conference on {H}ybrid {S}ystems: computation and control},	
	YEAR = {2012},
	PAGES = {227--236}, 
	MONTH = {April}, 
	ADDRESS = {Beijing, PRC},
}

@article{TMKA17,
title = "Quantitative model-checking of controlled discrete-time {M}arkov processes",
journal = "Information and Computation",
volume = "253",
pages = "1 - 35",
year = "2017",
issn = "0890-5401",
author = "Ilya Tkachev and Alexandru Mereacre and Joost-Pieter Katoen and Alessandro Abate",
}

@article{HS_TAC19,
  title={Robust dynamic programming for temporal logic control of stochastic systems},
  author={Haesaert, Sofie and Soudjani, Sadegh},
  journal={IEEE Transactions on Automatic Control},
  volume={66},
  number={6},
  pages={2496--2511},
  year={2020},
  publisher={IEEE}
}

@article{lavaei2021automated,
  title={Automated verification and synthesis of stochastic hybrid systems: A survey},
  author={Lavaei, Abolfazl and Soudjani, Sadegh and Abate, Alessandro and Zamani, Majid},
  journal={Automatica},
  volume={146},
  pages={110617},
  year={2022},
  publisher={Elsevier}
}

@inproceedings{HNS20_MultiObjective,
  title={Formal multi-objective synthesis of continuous-state MDPs},
  author={Haesaert, Sofie and Nilsson, Petter and Soudjani, Sadegh},
  booktitle={2021 American Control Conference (ACC)},
  pages={3428--3433},
  year={2021},
  organization={IEEE}
}

@ARTICLE{SA13,
	AUTHOR = {S. Soudjani and A. Abate},
	TITLE = {Adaptive and Sequential Gridding Procedures for the Abstraction and Verification of Stochastic Processes},
	JOURNAL = {SIAM Journal on Applied Dynamical Systems},
	NUMBER = {2},
	VOLUME = {12},
	YEAR = {2013},
	PAGES = {921--956},
}

@ INCOLLECTION{SAH12,
	AUTHOR = {S. Soudjani and A. Abate},
	TITLE = {Higher-Order Approximations for Verification of Stochastic Hybrid Systems},
	BOOKTITLE = {Automated Technology for Verification and Analysis},
	SERIES = {LNCS},
	VOLUME = {7561}, 
	YEAR = {2012},
	PAGES = {416--434},
	PUBLISHER = {Springer Verlag, Berlin Heidelberg},
 }

@article{SA14TCL,
  title={Aggregation and control of populations of thermostatically controlled loads by formal abstractions},
  author={Soudjani, Sadegh and Abate, Alessandro},
  journal={IEEE Transactions on Control Systems Technology},
  volume={23},
  number={3},
  pages={975--990},
  year={2014},
  publisher={IEEE}
}

@article{Coogan20,
  title={Abstraction-based synthesis for stochastic systems with omega-regular objectives},
  author={Dutreix, Maxence and Huh, Jeongmin and Coogan, Samuel},
  journal={Nonlinear Analysis: Hybrid Systems},
  volume={45},
  pages={101204},
  year={2022},
  publisher={Elsevier}
}

@PHDTHESIS{Luca98,
	AUTHOR = {L. {de Alfaro}},
	TITLE = {Formal Verification of Probabilistic Systems},
	SCHOOL = {Department of Computer Science, Stanford University},
	YEAR = {1998}
}

@book{BK08,
	AUTHOR = {C. Baier and J.-P. Katoen},
	PUBLISHER = {MIT Press},
	TITLE = {Principles of Model Checking},
	YEAR = {2008},
}

@inproceedings{piterman2006faster,
  title={Faster solutions of {R}abin and {S}treett games},
  author={Piterman, Nir and Pnueli, Amir},
  booktitle={21st Annual IEEE Symposium on Logic in Computer Science (LICS'06)},
  pages={275--284},
  year={2006},
  organization={IEEE}
}

@inproceedings{EJ91,
author = {Emerson, E. A. and Jutla, C. S.},
title = {Tree Automata, {Mu}-Calculus and Determinacy},
year = {1991},
publisher = {IEEE Computer Society},
address = {USA}, 
booktitle = {Proceedings of the 32nd Annual Symposium on Foundations of Computer Science}, pages = {368--377},
numpages = {10},
location = {San Juan, Puerto Rico},
series = {SFCS '91}
}

@inproceedings{long1994improved,
  title={An improved algorithm for the evaluation of fixpoint expressions},
  author={Long, David E and Browne, Anca and Clarke, Edmund M and Jha, Somesh and Marrero, Wilfredo R},
  booktitle={International Conference on Computer Aided Verification},
  pages={338--350},
  year={1994},
  organization={Springer}
}

@inproceedings{Thomas95,
  author    = {Wolfgang Thomas},
  title     = {On the Synthesis of Strategies in Infinite Games},
  booktitle = {{STACS} 95, 12th Annual Symposium on Theoretical Aspects of Computer
               Science, Munich, Germany, March 2-4, 1995, Proceedings},
  pages     = {1--13},
  year      = {1995},
  series    = {Lecture Notes in Computer Science},
  volume    = {900},
  publisher = {Springer},
}

@inproceedings{weininger2019satisfiability,
  title={Satisfiability bounds for $\omega$-regular properties in bounded-parameter {M}arkov decision processes},
  author={Weininger, Maximilian and Meggendorfer, Tobias and K{\v{r}}et{\'\i}nsk{\`y}, Jan},
  booktitle={2019 IEEE 58th Conference on Decision and Control (CDC)},
  pages={2284--2291},
  year={2019},
  organization={IEEE}
}

@article{Condon92,
  title={The complexity of stochastic games},
  author={Condon, Anne},
  journal={Information and Computation},
  volume={96},
  number={2},
  pages={203--224},
  year={1992},
  publisher={Elsevier}
}
